\documentclass[12pt]{article}
\usepackage[top=1in,bottom=1in,left=1in,right=1in]{geometry}
\usepackage{graphicx}
\usepackage{amsmath}
\usepackage{bbm}
\usepackage{amssymb}
\usepackage{amsthm}
\usepackage{wrapfig}
\usepackage{multicol}
\usepackage{float}
\usepackage{pdfpages}
\usepackage[colorlinks=true, urlcolor=blue, linkcolor=blue,citecolor=blue]{hyperref}
\usepackage{natbib}
\usepackage[utf8]{inputenc}
\usepackage{subfiles}
\usepackage{blindtext}
\usepackage{titlesec}
\usepackage{lipsum}
\usepackage[title,toc,titletoc]{appendix}
\usepackage{etoolbox}
\usepackage{tikz}
\usepackage[font=footnotesize, textfont=it]{caption}

\newcommand{\RN}[1]{  \textup{\uppercase\expandafter{\romannumeral#1}}}

\newtheorem{claim}{Claim}

\newtheorem{definition}{Definition}
\newtheorem{proposition}{Proposition}
\newtheorem{theorem}{Theorem}
\newtheorem{corollary}{Corollary}
\newtheorem{lemma}{Lemma}

\newtheorem{assumption}{Assumption}

\newtheorem{remark}{Remark}

\DeclareMathOperator*{\argmin}{argmin}

\DeclareMathOperator*{\diag}{diag}

\theoremstyle{definition} 
\newtheorem{example}{Example}

\title{Overabundant Information and Learning Traps\footnote{We thank Vasilis Syrgkanis for insightful comments in early conversations about this project. We are also grateful to Aislinn Bohren, Ben Golub, Paul Milgrom, Andrew Postlewaite, Ilya Segal, Carlos Segura, Rajiv Sethi, Andrzej Skrzypacz, and Yuichi Yamamoto for comments that improved this paper.}}
\author{Annie Liang\thanks{University of Pennsylvania} \quad \quad Xiaosheng Mu\thanks{Harvard University}}

\begin{document}
\maketitle

\begin{abstract}
We develop a model of social learning from overabundant information: Short-lived agents sequentially choose from a large set of (flexibly correlated) information sources for prediction of an unknown state. Signal realizations are public.  We demonstrate two starkly different long-run outcomes: (1) efficient information aggregation, where the community eventually learns as fast as possible; (2) ``learning traps," where the community gets stuck observing suboptimal sources and learns inefficiently. Our main results identify a simple property of the signal correlation structure that separates these outcomes. In both regimes, we characterize which sources are observed in the long run and how often.
\end{abstract}

\section{Introduction}
In many learning problems, agents cannot design their information in a completely flexible way. Instead, they choose from a given, finite (though often large) set of information sources. For instance, a researcher studying depression cannot access arbitrarily precise signals about the  role of stress. He can however acquire many kinds of information related to this question; for example, he might acquire neurochemical and genetic data from individuals undergoing stressful life changes, or collect observational data about simultaneous occurrences of depression and stress.

The value to acquiring new information depends on how that information relates to what is already known. Thus, past information acquisitions can shape the perception of which kinds of information are currently most useful. For example, a better understanding of the relationship between serotonin and depression increases the value to comparing serotonin levels across individuals (as opposed to other neurochemicals whose roles are less well-understood). Models of information acquisition often avoid explicitly describing these informational complementarities, as they complicate the analysis.\footnote{Exceptions include \citet{Borgers}, \citet{ChenWaggoner}, and \citet{Chade} among others.} But relationships across diverse sources  can have significant implications for behavior, in particular in dynamic learning environments where information is passed down over time.

The main contribution of this paper is to identify a new externality driven by complementarities across sources, and to characterize the consequences of this externality for long-run aggregation of information. We show that past information acquisitions have the potential to shape future acquisitions in two starkly different ways.
\begin{itemize}
\item \emph{Efficient information aggregation:} Past information helps agents to identify the ``best" kinds of information. At all sufficiently late periods, a social planner cannot improve on the history of acquisitions.
\item \emph{Learning traps:} Past information pushes agents to acquire information that leads to inefficiently slow learning. Early suboptimal choices propagate over time.
\end{itemize}
Relationships across the entire \emph{set} of information sources are relevant to which of these outcomes emerges, and our main results identify a key property that determines the outcome. 

In our model, agents are indexed by (discrete) time and sequentially choose from a large number of information sources, each  associated with a signal about a payoff-relevant state. We allow for flexible correlation across the sources by modeling each kind of information as a (noisy) linear combination of the payoff-relevant state and a set of ``confounding" variables. After acquiring information, each agent predicts the payoff-relevant state.

In contrast to the classic sequential learning model \citep{Hirshleifer,Banerjee,SmithSorenson}, we suppose that signal realizations are public, so that predictions are based on the history of realizations thus far. This departure permits us to focus on the externalities created by agents' choice of \emph{kind} of information, as opposed to the more frequently studied frictions that emerge from inference. 

We are mainly interested in settings with many sources of information, including some that are redundant. Formally, agents can completely learn the payoff-relevant state from exclusive (repeated) observation of signals from each of several possible subsets. As a benchmark, we first derive the \emph{optimal} long-run frequency of signal acquisitions. These correspond to the choices that maximize information revelation about the payoff-relevant state, and also to the choices that maximize a discounted sum of agent payoffs (in a patient limit).

Our main results demonstrate that whether society's acquisitions converge to this optimal long-run frequency depends critically on how many signals are needed to identify the payoff-relevant state. The key intuition refers back to an observation made in \citet{SethiYildiz}: An agent who repeatedly observes a source confounded by an unknown parameter learns \emph{both} about the payoff-relevant state and also about the confounding term, and hence improves his interpretation of \emph{this} source over time. In our setting, where a single confounding term can affect multiple sources, there is a further spillover effect: Learning from one source helps agents to interpret information from \emph{all} sources confounded by the same parameters. 

Suppose that in order to learn the payoff-relevant state, agents must observe a set of sources that reveals all of the confounding terms as well. Then endogenously, agents will acquire information that (collectively) reveals all of the unknowns. This will lead agents  to evaluate \emph{all} sources by an ``objective" asymptotic criterion, which then reveals the best set of sources. More formally, we obtain the following result: If $K$ sources are required to recover the payoff-relevant state (where $K$ is also the number of unknown states), then long-run acquisitions are optimal starting from any prior belief. 

In contrast, if it is possible to learn the payoff-relevant state without recovering all of the confounding terms, then agents can persistently undervalue sources that provide information confounded by these remaining variables. Thus, long-run learning may be inefficient.  Our second main result says that any set of fewer than $K$ sources that recovers the payoff-relevant state creates a ``learning trap" under some set of prior beliefs. We further show that the long-run inefficiency under a learning trap\textemdash measured as the ratio of optimal aggregated payoffs and achieved aggregated payoffs (under a social planner criterion)\textemdash can be arbitrarily large. 

The basic friction here is that investment in learning about confounding terms is socially beneficial, but not necessarily optimal for individuals.  Our main results show that this wedge between individual incentives and social objectives does not preclude long-run efficient learning. When the available kinds of information are related in certain ways, individual incentives will nevertheless endogenously drive individuals to acquire information in a way that is socially efficient. 

In the remaining cases, interventions may be needed to transition agents towards better sets of sources. In the final part of our paper, we study possible such interventions. We show that policymakers can restore efficient information aggregation by providing certain kinds of free information (that we characterize), or by reshaping the reward structure so that agents' predictions are based on information that they acquire over many periods. The success of these interventions depends on specific features of the informational environment, as we will discuss.

\subsection{Related Literature}\label{sec:relatedLit}
A recent literature considers choice from different kinds of information sources \citep{SethiYildiz,CheMierendorff,FudenbergStrackStrzalecki,Mayskaya,LiangMuSyrgkanis,SethiYildiz2}. We build upon \cite{LiangMuSyrgkanis} in particular, which introduced the framework we use here (see Section \ref{sec:model}) under the restriction that all of the sources are required to learn the payoff-relevant state. This restriction rules out the possibility of overabundant information, which is the focus of the present paper. 

\citet{SethiYildiz,SethiYildiz2} study long-run (myopic) acquisitions from a large number of Gaussian sources, as we do. Our model differs from this work in a few key ways: First, \citet{SethiYildiz,SethiYildiz2} consider stochastic error variances, so that the ``best" sources vary from period to period; in contrast, we fix error variances, so there is (generically) a unique ``best" asymptotic set. Second, \citet{SethiYildiz,SethiYildiz2} focus on correlation structures that fall under our ``learning traps" result (part (a) of Theorem \ref{thm:general}), while we explore arbitrary correlation structures and show that many  lead to optimal learning.

Our model builds on the social learning and herding literatures \citep{SmithSorenson,Banerjee,Hirshleifer}, which consider information aggregation by short-lived agents who sequentially acquire information.  At a high level, the externality identified in our paper relates to the classic externality from this literature: In both settings, the precision of public information can grow inefficiently slowly because of endogenous information acquisitions driven by past choices. But in the present paper, all signal realizations are publicly and perfectly observed, which turns off the inference problem  essential to the existence of cascades in standard herding models. Our focus is on a new mechanism, in which externalities arise through choice of \emph{kind} of information; as we will see, this externality has a rather different structure.

Our setting with choice of information connects to the work by \citet{BurguetVives}, \cite{FrankPai}, and \citet{Ali}, which introduced endogenous information acquisition to social learning models. Relative to this work, our paper considers  choice from a fixed set of information sources (with a capacity constraint), in contrast to choice from a flexible set of information sources (with a cost on precision). Our results focus on the \emph{speed} of learning, as in \citet{Vives}, \citet{GolubJackson}, \citet{Tamuz}, and \citet{StrackTamuz} among others.

Our motivation of the optimal frequency in Section \ref{sec:optimalAsymp} is related to the experimental design literature in statistics, and in particular to the notion of $\mathbf{c}$-optimality (choice of $t$ experiments to minimize the posterior variance of an unknown state). \citet{Chaloner} showed that a $c$-optimal design exists on at most $K$ points. Our Theorem \ref{thm:totalOpt} extends this result, supplying a characterization of the optimal design itself and demonstrating uniqueness.\footnote{Another difference is that \citet{Chaloner} studies the optimal continuous design, while we impose an integer constraint on signal counts.} We further show that the optimal frequency maximizes a discounted objective in the patient limit. Finally, our main results admit a re-interpretation that connects to a literature on learning convergence in potential games, with more detail given in Section \ref{sec:intuition}.

\section{Framework}\label{sec:model}
There are $K$ persistent unknown states: a payoff-relevant state $\omega$ and $K-1$ confounding states $b_1, \dots, b_{K-1}$. We assume that the state vector $\theta:=(\omega, b_1,\dots,b_{K-1})'$ follows a multivariate normal distribution $\mathcal{N}(\mu^0, \Sigma^0)$,\footnote{All vectors in this paper are column vectors.} where the prior covariance matrix $\Sigma^0$ has full rank.\footnote{The full rank assumption is without loss of generality: If there is linear dependence across the states, the model can be mapped into an equivalent setting that has a lower dimensional state space and that satisfies the full rank condition.}

Agents have access to $N$ \emph{kinds} or \emph{sources} of information. Observation of source $i$ produces an independent realization of the random variable
\[
X_i = \langle c_i, \theta \rangle + \epsilon_i, \quad \epsilon_i \sim \mathcal{N}(0,1)
\]
where $c_i = (c_{i1}, \dots, c_{iK})'$ is a vector of constants, and the error terms $\epsilon_i^t$ are independent from each other and over time. It is without loss to normalize the error terms as we do above, since the coefficients $c_i$ are unrestricted (thus, signals can be of differing precision levels). Throughout, we take $C$ to be the $N \times K$ coefficient matrix whose $i$-th row is $c_i'$. 

The payoff-irrelevant states produce correlations across the sources, and we can interpret these states for example as:
\begin{itemize}
\item \emph{Confounding explanatory variables}: Observation of signal $i$ produces the (random) outcome $y= \omega c^1_i + b_1 c^2_i + \dots b_{K-1} c^K_i +\epsilon_i$, which depends linearly on an observable characteristic vector $c_i$. For example, $y$ might be the average incidence of depression in a group of individuals with characteristics $c_i$. The state of interest $\omega$ is the coefficient on a given characteristic $c^1_i$, and the payoff-irrelevant states are the unknown coefficients on the auxiliary characteristics. Different sources represent subpopulations with different characteristics.

\item \emph{Knowledge and technologies that aid interpretation of information}: Interpret each signal as a measurement. For example, researchers studying depression can acquire measurements of various neurochemicals from lab subjects. Neurochemicals differ in how precisely they can be measured (measurement error) and how well their role is understood (precision of interpretation). The confounding terms represent the quality of these different measurement technologies, and also the degree of background contextual knowledge.
\end{itemize}

Agents indexed by (discrete) time $t \in \mathbb{N}$ move sequentially. Each agent $t$ chooses one of the $N$ signals and observes an independent realization. He then predicts the state $\omega$, receiving a payoff of $-(a-\omega)^2$. We assume throughout that all signal realizations are public; thus, predictions are based on the entire history of signal acquisitions and realizations.

 Throughout, the set of signals is indexed by $[N] = \{1, \dots, N\}$. We call a subset of signals $\mathcal{S} \subset [N]$ \emph{spanning} if the vectors $\{c_i: i \in \mathcal{S}\}$ span the coordinate vector $e_1=(1,0,\dots,0)' \in \mathbb{R}^{K}$, so that it is possible to learn the payoff-relevant state $\omega$ by repeatedly observing signals from only $\mathcal{S}$. We call $\mathcal{S}$ \emph{minimally spanning} if it is spanning, and no proper subset of $\mathcal{S}$ is spanning.
 
We assume in this paper that the complete set of signals $[N]$ is spanning, so that the payoff-relevant state can be recovered by observing all signals infinitely often.\footnote{This assumption is without loss, and our results do extend to situations where $\omega$ is \emph{not} identified from the available signals. To see this, we first take a linear transformation and work with the following equivalent model: The state vector $\tilde{\theta}$ is $K$-dimensional \emph{standard Gaussian}, each signal $X_i = \langle \tilde{c_i}, \tilde{\theta} \rangle + \epsilon_i$, and the payoff-relevant parameter is $\langle u, \tilde{\theta} \rangle$ for some fixed vector $u$. Let $R$ be the subspace of $\mathbb{R}^K$ spanned by $\tilde{c}_1, \dots, \tilde{c}_N$. Then project $u$ onto $R$: $u = r + w$ with $r \in R$ and $w$ orthogonal to $R$. Thus $\langle u, \tilde{\theta} \rangle = \langle r, \tilde{\theta} \rangle + \langle w, \tilde{\theta} \rangle$. By assumption, the random variable $\langle w, \tilde{\theta} \rangle$ is independent from any random variable $\langle c, \tilde{\theta} \rangle$ with $c \in R$ (because they have zero covariance). Thus the uncertainty about $\langle w, \tilde{\theta} \rangle$ cannot be reduced upon any signal observation. Consequently, agents only seek to learn about $\langle r, \tilde{\theta} \rangle$, returning to the case where the payoff-relevant parameter \emph{is} identified.} This assumption nests two interesting cases. Say that the informational environment has \emph{exactly sufficient information} if $[N]$ is minimally spanning.  Then, it is possible to recover $\omega$ by observing each information source infinitely often, but not by exclusively observing any proper subset of sources. 

We are primarily interested in settings of \emph{informational overabundance}, where $[N]$ is spanning but not minimally spanning. Multiple different subsets of signals allow for recovery of $\omega$, and a key point of our analysis is to compare the set of sources that ``should" be observed in the long run with the set of sources that is in fact observed in the long run. Except for trivial cases, informational overabundance corresponds to $N>K$ (more signals than states).\footnote{It is possible for $\omega$ to be ``overidentified" from a set of $N\leq K$ signals, e.g.\ $X_1 = \omega + \epsilon_1$, $X_2=\omega + b_1 + b_2 + \epsilon_2$, and $X_3 = b_1 + b_2 +\epsilon_3$. In this case, the set $\{X_1,X_2,X_3\}$ is spanning, but not minimally spanning since both of its subsets $\{X_1\}$ and $\{X_2,X_3\}$ are also spanning. Although $N=K=3$ in this example, it is equivalent to a model in which there is a single confounding term $\tilde{b}_1= b_1 + b_2$, and the three signals are rewritten $X_1 = \omega + \epsilon_1$, $X_2 = \omega + \tilde{b}_1 + \epsilon_2$ and $X_3 =  \tilde{b}_1 + \epsilon_3$. Then we do have $N>K$ in this equivalent model.}

\section{Preliminaries}\label{sec:prelim}

Each agent $t$ faces a history $h^{t-1} \in ([N] \times \mathbb{R})^{t-1} = H^{t-1}$ consisting of all past signal choices and their realizations. Write $\theta \sim \mathcal{N}(\mu^{t-1}, \Sigma^{t-1})$ for the agent's beliefs about the state vector, prior to making his own signal choice. Given an observation of signal $i$, his posterior beliefs become $\theta \sim \mathcal{N}(\mu^t, \Sigma^t)$. The posterior covariance matrix $\Sigma^t$ is a \emph{deterministic} function of the prior covariance matrix $\Sigma^{t-1}$ and the signal choice $X_i$, and the posterior mean is the random vector $\mu^t \sim \mathcal{N}(\mu^{t-1}, \Sigma^{t-1} - \Sigma^t)$. 

Agent $t$'s posterior belief about \emph{the payoff-relevant state} $\omega$ (after observing his signal)  is given by $\omega \sim \mathcal{N}(\mu^t_1,\Sigma^t_{11})$,\footnote{Subscripts indicate particular entries of a vector or matrix.} and his maximum expected payoff (corresponding to prediction of his posterior mean) is $-\Sigma^t_{11}$. Thus, the signal acquisition that maximizes agent $t$'s payoffs is the one that minimizes his posterior variance about $\omega$. 

We can track society's acquisitions as a sequence of division vectors 
\[m(t)=(m_1(t), \dots, m_N(t))\]
where $m_i(t)$ is the number of times that signal $i$ has been observed up to and including time $t$. Let $V(q_1, \dots, q_N)$ denote posterior variance about $\omega$, given the initial prior covariance matrix $\Sigma^0$ and $q_i$ observations of each signal $i$.\footnote{For a normal prior and normal signals, the posterior covariance matrix does not depend on signal realizations. See Appendix \ref{appx:prelim} for the complete (closed-form) expression for $V$.} Then, $m(t)$ evolves deterministically according to the following rule: $m(0)$ is the zero vector, and for each time $t \geq 0$ and signal $i$, 
\[
m_i(t+1) = \left\{ \begin{array}{cl} m_i(t) + 1  & \text{if }  V(m_i(t)+1, m_{-i}(t)) \leq V(m_j(t)+1, m_{-j}(t))~ \forall j. \\
 m_i(t) & \mbox{otherwise.} \end{array} \right.
\]
That is, in each period $t$ the division vector increases by $1$ in exactly one coordinate, corresponding to the signal that allows for the greatest immediate reduction in posterior variance. We allow ties to be broken arbitrarily, so there may be multiple possible paths $m(t)$. We are interested in the \emph{long-run frequencies} of observation $\lim_{t \rightarrow \infty} m_i(t)/t$ for each signal $i$. This describes the fraction of periods that is (eventually) devoted to each signal.\footnote{We show in Section \ref{sec:mainResults} that the limit exists under a mild technical assumption.} Note the possibility for some signals to have zero long-run frequency.

\section{Optimal Benchmark} \label{sec:optimalAsymp}
We begin by characterizing the optimal frequency with which each signal \emph{should} be viewed in the long run, both from a perspective of information revelation and also from the perspective of a social planner aggregating agent payoffs. These frequencies will serve as a benchmark against which to evaluate the actual social acquisitions.

\subsection{Definition and Motivation}
Consider the one-shot problem in which some $t$ observations can be allocated across the available signals, and the goal is to minimize posterior variance about $\omega$. Call any solution 
\[ 
n(t) \in \argmin_{(q_1, \dots,  q_K): q_i \in \mathbb{Z}^{+}, \sum_{i} q_i = t} V(q_1, \dots, q_K)
\]
a \emph{$t$-optimal division vector}. Signal acquisitions prior to time $t$ maximize information revelation about $\omega$ if and only if their composition is described by $n(t)$. Thus, we will refer to $n(t)$ as the \emph{optimal count} over signals.

Define $\lambda^{OPT}=\lim_{t\rightarrow \infty} n(t)/t$ to be the limiting frequency over signals corresponding to these $t$-optimal division vectors; that is, $\lambda^{OPT}_i$ is the fraction of periods that is (asymptotically) optimally devoted to signal $i$. These frequencies are well-defined and unique under a subsequent condition (Assumption \ref{assumptionUniqueMin}), and we will refer to $\lambda^{OPT}$ from here on as the \emph{optimal (long-run) frequency vector}.

The above justification for $\lambda^{OPT}$ is purely in terms of information revelation. The frequency vector $\lambda^{OPT}$ also emerges, however, as an approximation of the (limiting patient) solution to a social planner problem, in which signals and actions are chosen to maximize a discounted aggregate of payoffs across individuals:
 \[
U_{\delta} := \mathbb{E} \left[\sum_{t = 1}^{\infty} \delta^{t-1} \cdot u(a_t, \omega) \right]. 
\]
For fixed $\delta$, let $d_{\delta}(t)$ be the vector of signal counts (up to period $t$) associated with any strategy that maximizes $U_{\delta}$. The long-run frequency vector produced by this optimal sampling procedure is $\lim_{t\rightarrow \infty} d_{\delta}(t)/t$, and the result below says that $\lambda^{OPT}$ approximates this long-run frequency vector in the patient limit $\delta \rightarrow 1$.

\begin{proposition}\label{prop:highDelta}
Under the subsequent Assumption \ref{assumptionUniqueMin}, for any $\epsilon > 0$, there exists $\underline{\delta} < 1$ such that for any $\delta \geq \underline{\delta}$ it holds that
\[
\limsup_{t \rightarrow \infty} ~\left\| \frac{d_{\delta}(t)}{t} - \lambda^{OPT} \right\| \leq \epsilon.
\]
Here $\lvert \lvert \cdot \rvert \rvert$ represents the Euclidean norm. 
\end{proposition}

This result further justifies our interest in $\lambda^{OPT}$.\footnote{It is clear that the strategy that samples signals (randomly) according to $\lambda^{OPT}$ is best among \emph{stationary} information acquisition strategies as $\delta \rightarrow 1$. This is however weaker than our result, because the optimal strategy for any fixed $\delta$ may be far from stationary.} Throughout, we will say that \emph{efficient information aggregation} occurs if society's frequency of observations tends to $\lambda^{OPT}$. 

\subsection{Characterization}
The main difficulty in characterizing the optimal frequency vector $\lambda^{OPT}$ is that we do not know its support. Because $\omega$ is recoverable from different proper subsets of signals, there are many signals that could conceivably receive zero frequency in the long-run.

As a first step, we thus consider a simpler version of the social planner problem in which  agents can acquire signals from (only) some minimal spanning set $\mathcal{S}$. By definition, the setting is one of exactly sufficient information, where all signals must be observed in order to recover $\omega$. It is straightforward to show that the first coordinate vector $e_1$ can be decomposed as a (unique) linear combination of signals in $\mathcal{S}$:
\[
e_1=\sum_{i \in \mathcal{S}} \beta_i^\mathcal{S} \cdot c_i
,\] 
where the coefficients $\beta_i^{\mathcal{S}}$ are non-zero.

We showed in prior work that  each signal $i \in \mathcal{S}$ should be observed (asymptotically)  in proportion to its coefficient $\beta_i^\mathcal{S}$:
   
\begin{proposition}[\citet{LiangMuSyrgkanis}]\label{prop:N=K}
Suppose agents are constrained to a minimal spanning set $\mathcal{S}$. Then, for every signal $i \in \mathcal{S}$, the optimal count satisfies
\begin{equation}\label{eq:asympFreq}
n_i^{\mathcal{S}}(t) = \frac{\lvert \beta^{\mathcal{S}}_i \rvert}{\sum_{j \in \mathcal{S}} \lvert \beta^{\mathcal{S}}_j \rvert} \cdot t + O(1)
\end{equation}
where $O(1)$ represents a residual term that remains bounded as $t \rightarrow \infty$.

\end{proposition}
\noindent 
Thus, the optimal long-run frequency with which signal $i$ is viewed is $\lvert \beta^{\mathcal{S}}_i \rvert / \left(\sum_{j \in \mathcal{S}} \lvert \beta^{\mathcal{S}}_j \rvert \right)$.
\bigskip

To understand these critical coefficients $\beta_i^\mathcal{S}$, consider (for simplicity) the case in which the set $\mathcal{S}$ has size $K$. The coefficient vectors associated with signals in $\mathcal{S}$ have full rank,\footnote{Otherwise, $\mathcal{S}$ would not be a \emph{minimal} spanning set.} and we let $C_{\mathcal{S}}$ denote the matrix of these coefficient vectors. The (random) vector of realizations corresponding to one observation of each signal in this set can be written as
\[Y = (y_1, \dots, y_K)' = C_{\mathcal{S}}\theta + \varepsilon\]
where  $\varepsilon$ is a $K \times 1$ vector of error terms. Given these realizations, the best linear unbiased estimate for $\omega$ is $\hat{\omega}= \left[C_{\mathcal{S}}^{-1} Y\right]_{11}.$
Perturbing the realization of signal $i$ by $\delta_i$ changes this estimate by $[C_{\mathcal{S}}^{-1}]_{1i}\cdot \delta_i$. One can show that the coefficients $\beta_i^{\mathcal{S}}= \lvert [C_{\mathcal{S}}^{-1}]_{1i} \rvert$, so the larger $\beta_i^{\mathcal{S}}$ is, the more $\hat{\omega}$ responds to changes in the realization of signal $i$. Proposition \ref{prop:N=K} thus says that agents should observe more frequently those signals whose realizations more strongly influence the best linear estimate of $\omega$. 

This proposition additionally implies the following corollary regarding the speed of learning from $\mathcal{S}$: 

\begin{corollary}
Under optimal sampling from any minimal spanning set $\mathcal{S}$, the minimum achievable posterior variance after $t$ observations satisfies: 
\[
V\left(n^{\mathcal{S}}(t)\right) \sim \phi(\mathcal{S})^2/t:= \left(\sum_{i \in \mathcal{S}} \lvert \beta^{\mathcal{S}}_i \rvert \right)^2/t.
\]
where the notation ``$F(t) \sim G(t)$" means $\lim_{t \rightarrow \infty} \frac{F(t)}{G(t)} = 1$. 
\end{corollary}


In what follows, we work with the simpler statistic $\phi(\mathcal{S})$ (roughly an \emph{asymptotic standard deviation}), noting that the smaller $\phi(\mathcal{\mathcal{S}})$ is, the faster the community learns from $\mathcal{S}$. We assume throughout that there is a \emph{best} minimal spanning set according to $\phi$:

\begin{assumption}[Unique Minimizer]\label{assumptionUniqueMin}
$\phi(\mathcal{S})$ has a unique minimizer $\mathcal{S}^*$ among minimal spanning sets $\mathcal{S} \subset [N]$. 
\end{assumption}

\noindent This assumption is a restriction on the coefficient matrix $C$, and it rules out examples such as the following:


\begin{example} \label{ex:duplicate2} 
The available signals are
$X_1 = \omega + b_1+ \epsilon_1$,
$X_2 = b_1 +\epsilon_2$,
$X_3 = \omega + b_2 + \epsilon_3$, and
$X_4 = b_2 + \epsilon_4$.
Assumption \ref{assumptionUniqueMin} fails, because learning occurs equally fast from either of the minimal spanning sets $\{X_1,X_2\}$ and $\{X_3,X_4\}$.
\end{example}

\noindent  We note that Assumption \ref{assumptionUniqueMin} holds generically\textemdash for example, it holds given arbitrarily small perturbations of the above environment.\footnote{Throughout the paper, ``generic" means with probability $1$ for signal coefficients $c_{ij}$ randomly drawn from a full support distribution on $\mathbb{R}^{NK}$.}

\bigskip

If we restrict agents to sample exclusively from a single minimal spanning set, then the optimal sampling rule (under Assumption \ref{assumptionUniqueMin}) is clearly the frequency vector $\lambda^* \in \Delta^{N-1}$ satisfying
\begin{equation}\label{eq:lambda*}
\lambda^*_i = \left\{\begin{array}{cl}
\frac{\lvert \beta_i^{\mathcal{S}^*} \rvert }{\sum_{j \in \mathcal{S}^*} \lvert \beta_j^{\mathcal{S}^*} \rvert} &\quad \forall \,\, i \in \mathcal{S}^* \\
 0& \quad \forall \,\, i \notin \mathcal{S}^*   \end{array}\right.
\end{equation}
This sampling rule assigns zero frequency to signals outside of the set $\mathcal{S}^*$, and samples signals within $\mathcal{S}^*$ according to the frequencies given in Proposition \ref{prop:N=K}. 

In principle, a social planner may improve on $\lambda^*$ by sampling from multiple spanning sets. Our first theorem shows to the contrary that $\lambda^*$ remains optimal when arbitrary sampling procedures are permitted. So long as $C$ satisfies Unique Minimizer, the best long-run strategy is to restrict to the best minimal spanning set and sample from that set as in Proposition \ref{prop:N=K}.

\begin{theorem}\label{thm:totalOpt}
Let $\lambda^*$ be given by (\ref{eq:lambda*}). Under Unique Minimizer, the optimal long-run frequency vector satisfies $\lambda^{OPT}=\lambda^*$.\footnote{This theorem equivalently states that $n_i(t) \sim \lambda^*_i \cdot t$. In Appendix \ref{appx:n(t)finite}, we show the stronger result that $n_i(t) = \lambda^*_i \cdot t + O(1)$.}
\end{theorem}

\noindent The conclusion can be loosely interpreted as stating that $\lambda^*$ is the ``most efficient linear representation'' of the payoff-relevant state in terms of the signal coefficients.\footnote{Specifically, consider the following constrained minimization problem: 
$\min \sum_{i = 1}^{N} \lvert \beta_i \rvert$ subject to $\sum_{i=1}^{N} \beta_i \cdot c_i = e_1.$
It can be shown by linear programming that the minimum is attained exactly when $\beta_i = \beta^{\mathcal{S}^*}_i$ (that is, when focusing on a single minimal spanning set).}

We show in Appendix \ref{appx:examples} that Assumption \ref{assumptionUniqueMin} is necessary for this result: In the environment described in Example \ref{ex:duplicate2}, there are priors given which it is strictly optimal to observe all four available signals with positive frequency.

\section{Main Results}\label{sec:mainResults}
In general, we may expect a difference between the best one-shot allocation of $t$ acquisitions, described in the previous section, and the set of $t$ acquisitions chosen by sequential (short-lived) decision-makers. We turn now to characterizing the latter. 

Sections \ref{sec:learningTraps} and \ref{sec:efficient} focus on welfare evaluation\textemdash that is, will society's acquisitions $m(t)$ eventually approximate the optimal acquisitions $n(t)$? We show that signal correlation structures can be classified into two kinds\textemdash those for which efficient information aggregation is guaranteed (starting from all prior beliefs), and those for which ``learning traps" are possible (depending on the prior belief). Separation of these two classes depends critically on how many signals are required to identify $\omega$.

For a given signal structure, the set of signals that are observed in the long run potentially depends on the prior belief. A related question is then \emph{which} long-run outcomes are possible for a given signal structure, when we allow for arbitrary priors. We provide a complete characterization in Section \ref{sec:general}.



\subsection{Learning Traps} \label{sec:learningTraps}
The following simple example demonstrates that sequential information acquisition need not lead to efficient information aggregation. Indeed, the set of signals that are observed in the long run can be disjoint from the optimal set.

\begin{example}\label{ex:LT}
There are three available signals:
\begin{align*}
X_1 &= \omega + \epsilon_1\\
X_2 &= 3\omega + b_1 +\epsilon_2\\
X_3 &= b_1 + \epsilon_3
\end{align*}
Both $\{X_1\}$ and $\{X_2,X_3\}$ are minimal spanning sets, but $\{X_2,X_3\}$ is optimal.\footnote{It is straightforward to verify that $\phi(\{X_1\}) = 1 > 2/3= \phi(\{X_2,X_3\})$. Note also that $X_2+X_3$ is an unbiased signal about $\omega$, and it is more informative than two realizations of $X_1$; this demonstrates that $\{X_2,X_3\}$ is the best minimal spanning set without direct computation of $\phi$.}

Consider a prior where $\omega$ and $b_1$ are independent, and the prior variance of $b_1$ is large (exceeds $8$). In the first period, observation of $X_1$ is most informative about $\omega$, since $X_2$ is perceived as a noisier signal \emph{about $\omega$} than $X_1$, and observations of $X_3$ provide information only about the confounding term $b_1$ (which is uncorrelated with $\omega$).  Thus the best choice is to observe $X_1$. This observation does not affect the variance of $b_1$, so the same argument shows that every agent observes signal $X_1$.\footnote{Note that the existence of learning traps is not special to the assumption of normality. We report a related example with non-normal signals in Appendix \ref{appx:non-normalLT}.} We refer to $\{X_1\}$ in this example as a \emph{learning trap}.
\end{example} 

Generalizing this example, the result below (stated as a corollary, since it will follow from the subsequent Theorem \ref{thm:general}) gives a sufficient condition for learning traps. We impose the following (generic) assumption on the informational environment, which requires that every set of $k \leq K$ signals is linearly independent:

\begin{assumption}[Strong Linear Independence] 
$N\geq K$ and every $K\times K$ submatrix of $C$ is of full rank.\footnote{Besides trivial cases with redundant signals, Strong Linear Independence also rules out settings such as the following:
$
X_1 = \omega + b_1+ \epsilon_1$, $X_2 = b_1 +\epsilon_2$,
$X_3 = 2\omega + b_2 + \epsilon_3$, $X_4 = b_2 + \epsilon_4$, 
$X_5 = 3\omega + b_3 + \epsilon_5$, and $X_6 = b_3 + \epsilon_6$.
Then $K = 4$ but the four signals $X_1, X_2, X_3, X_4$ are \emph{not} linearly independent.} 
 \label{ass:SLI}
\end{assumption}

\begin{corollary}\label{corMyopic}
Assume Strong Linear Independence. For every minimal spanning set $\mathcal{S}$ with $\vert \mathcal{S} \vert < K$, there exists an open set of prior beliefs given which agents exclusively observe signals from $\mathcal{S}$. 
\end{corollary}

\noindent Thus, every small set (fewer than $K$ signals) that identifies $\omega$ is a candidate learning trap. 

In special environments, simple bounds on the extent of inefficiency are possible. For example, if there is an unbiased signal $c\omega+\epsilon$, then the posterior variance at each time $t$ cannot exceed $1/(c^2t)$, and so aggregated payoffs must be at least $-\sum_{t\geq 1} \delta^{t-1}/(c^2 t)$.\footnote{This is because each agent can at least sample this unbiased signal and improve the posterior precision (i.e., inverse of the posterior variance) by $c^2$. We thank Andrzej Skrzypacz for this observation.}  The size of inefficiency cannot be uniformly bounded across environments, however. Specifically, for any positive number $L$, there exists a set of signals and a prior belief given which 
\[
\frac{\phi(\mathcal{S})}{\phi(\mathcal{S}^*)}> L
\]
where $\mathcal{S}$ is the set of signals observed in the long run with positive frequency, and $\mathcal{S}^*$ is the optimal set. This can be shown by direct construction: Modify the example above so that $X_2 = \alpha \omega + b_1 +\epsilon_2$, with $\alpha$ sufficiently large. For every choice of $\alpha$, there is a set of priors given which $X_1$ is again exclusively observed.\footnote{The region of inefficient priors (that result in suboptimal learning) does decrease in size as the level of inefficiency increases. As $\alpha$ increases, the prior variance of $b_1$ has to increase correspondingly in order for the first agent to choose $X_1$.} Thus, society's long-run speed of learning can be arbitrarily slower than the optimal speed. This also implies that\textemdash from the perspective of a social planner who maximizes $\delta$-discounted payoffs\textemdash the payoff ratio between society's sequential acquisitions and optimal acquisitions can be arbitrarily large as $\delta \rightarrow 1$.\footnote{To derive this payoff comparison, note that $\phi(S)/\phi(\mathcal{S}^*) > L$ implies the ratio of flow payoffs in any late period $t$ is larger than $L$. Using the fact that the harmonic series diverges, we know that as $\delta \rightarrow 1$, these (later) payoffs dominate the total payoffs from the initial periods.}

Finally, note that learning traps can emerge even if the best spanning set has the smallest size among all minimal spanning sets. We demonstrate this below.

\begin{example} The available signals are:
\begin{align*}
X_1 & =10\omega + b_1 + \epsilon_1 \\
X_2 &= b_1 + \epsilon_2 \\
X_3 &= \omega+b_2+\epsilon_3\\
X_4&=b_2 + b_3 + \epsilon_4\\
X_5&=b_3+\epsilon_5
\end{align*}
Here, the best set is $\{X_1,X_2\}$, which is also the minimal spanning set of lowest cardinality. But suppose the prior belief is such that $\omega, b_1, b_2, b_3$ are independent, and there is initially high uncertainty about $b_1$ and low uncertainty about $b_2$ and $b_3$. Then, agents suboptimally acquire only observations of $X_3$, $X_4$, and $X_5$.
\end{example}

\subsection{Efficient Information Aggregation} \label{sec:efficient}
Suppose in contrast to the previous section that repeated observation of $K$ sources is  required to recover $\omega$. Our next result shows that a very different long-run outcome obtains: Starting from \emph{any} prior, information acquisition eventually approximates the optimal frequency. Thus, even though agents are short-lived (``myopic"), they will end up acquiring information in a way that is socially best.

\begin{corollary}\label{corEfficient}
Under Unique Minimizer, if every minimal spanning set has size $K$, then starting from any prior belief, it holds that $m_i(t) \sim \lambda^*_i \cdot t$ for every signal $i$.\footnote{In Appendix \ref{appx:m(t)finite}, we further show $m_i(t) = \lambda^*_i(t) + O(1)$ holds in this case. That result implies the difference between society's acquisitions $m(t)$ and optimal acquisitions $n(t)$ remains bounded as $t$ increases.}
\end{corollary} 

\noindent A detailed intuition for this result appears in the subsequent Section \ref{sec:intuition}.

The condition that all minimal spanning sets have size $K$ is generically satisfied.\footnote{We point out that the set of coefficient matrices satisfying Unique Minimizer is ``generic" in the following stronger sense: Fix the \emph{directions} of coefficient vectors, and suppose that the \emph{precisions} are drawn at random; then, generically different minimal spanning sets correspond to different speed of learning. In contrast, whether every minimal spanning set has size $K$ is a condition on the \emph{directions} themselves.} However, if we expect that sources are endogenous to design or strategic motivations, the relevant informational environments may not fall under this condition. For example, the existence of an unbiased signal about $\omega$ (that is, $X= c \omega + \epsilon$) is non-generic, but plausible in practice. Sets of signals that partition into different groups (with group-specific confounding terms) are also economically interesting but non-generic. The previous Corollary \ref{corMyopic} shows that inefficiency is a possible outcome in these cases. 

\subsection{Characterization of Long-Run Outcomes} \label{sec:general}
We now provide a complete characterization of the possible long-run observation sets for any environment. Here we need to consider \emph{subspaces spanned by different signal sets}. Formally, for any spanning set of signals $\mathcal{A}$, let $\overline{\mathcal{A}}\subseteq [N]$ be the set of available signals whose coefficient vectors belong to the subspace spanned by signals in $\mathcal{A}$. We say that a minimal spanning set $\mathcal{S}$ is \emph{subspace-optimal} if it uniquely maximizes the speed of learning among feasible sets of signals within its subspace.  
\begin{definition} \label{def:subspace}  A minimal spanning set $\mathcal{S}$ is \emph{subspace-optimal} if it uniquely minimizes $\phi$ among all subsets of $\overline{\mathcal{S}}$ that are minimally spanning. 
\end{definition} 

\begin{example} 
Suppose the available signals are $X_1 = \omega + \epsilon_1$ and $X_2 = 2\omega + \epsilon_2$, and define $\mathcal{S}=\{X_1\}$. Then, $\overline{\mathcal{S}}=\{X_1,X_2\}$. Since $\{X_2\}$ permits faster speed of learning than $\{X_1\}$, the set $\mathcal{S}$ is not subspace-optimal.
\end{example}

We introduce one final assumption, which strengthens Unique Minimizer to require the existence of a best minimal spanning set $\mathcal{S}$ within every subspace.

\begin{assumption}[Unique Minimizer in Every Subspace]\label{assumptionUniqueMinStrong}
For every $\mathcal{A} \subset [N]$, there exists a unique minimal spanning set $\mathcal{S}$ that minimizes $\phi$ among subsets of $\overline{\mathcal{A}}$. 
\end{assumption}
\noindent This assumption is guaranteed if different minimal spanning sets correspond to different $\phi$-values, and thus holds generically. 

Our next result generalizes both the learning traps result and also the efficient information aggregation result from the previous sections. Theorem \ref{thm:general} says that long-run information acquisitions eventually concentrate on a set $\mathcal{S}$ (starting from some prior belief) \emph{if and only if} $\mathcal{S}$ is a subspace-optimal minimal spanning set.

\begin{theorem}\label{thm:general}
(a) Suppose $\mathcal{S}$ is a subspace-optimal minimal spanning set. Then, there exists an open set of prior beliefs given which agents exclusively observe signals from $\mathcal{S}$. Long-run frequencies are positive precisely for those signals in $\mathcal{S}$, and they are given by Proposition \ref{prop:N=K}. 

(b) Under Assumption \ref{assumptionUniqueMinStrong}, long-run frequencies exist given any prior belief. Moreover, if $\mathcal{S}$ denotes the signals viewed with positive long-run frequency, then $\mathcal{S}$ is a minimal spanning set that is subspace-optimal, and the long-run frequencies are given by Proposition \ref{prop:N=K}.
\end{theorem}

This theorem directly implies our previous Corollaries \ref{corMyopic} and \ref{corEfficient}. To see this, observe that under Strong Linear Independence, $\overline{\mathcal{S}}=\mathcal{S}$ for every minimal spanning set $\mathcal{S}$ with fewer than $K$ signals.\footnote{Suppose $\lvert \mathcal{S} \rvert < K$, then by assumption of Strong Linear Independence, every signal not in $\mathcal{S}$ is linearly independent from the signals in $\mathcal{S}$. Hence $\overline{\mathcal{S}}$ cannot contain any other signal.}  This implies that every minimal spanning set with fewer than $K$ signals is (trivially) optimal in its subspace, producing Corollary \ref{corMyopic} from part (a) of the theorem.

On the other hand, if every minimal spanning set has size $K$, then all minimal spanning sets belong to the same subspace. Under Unique Minimizer, there can only be one minimal spanning set that is optimal in this subspace, and  this must also be the best set overall (in the sense of Section \ref{sec:optimalAsymp}). This yields Corollary \ref{corEfficient} from part (b) of the theorem above.

We collect below a few additional implications of Theorem \ref{thm:general}:

\begin{corollary} 
Under Unique Minimizer, if the best set $\mathcal{S}^*$ is of size 1 (equivalently, if $\mathcal{S}^*$ consists of a single unbiased signal $\alpha \omega + \epsilon$), then learning traps cannot emerge.\footnote{This is because the unbiased signal belongs to every subspace spanned by a minimal spanning set.}
\end{corollary}

\begin{corollary} 
Under Unique Minimizer, learning traps of size $K$ cannot emerge.\footnote{If $\mathcal{S}$ is a learning trap of size $K$, then it spans the whole space. But part (b) of Theorem \ref{thm:general} shows $\mathcal{S}$ must be the best set, leading to a contradiction.}
\end{corollary}

\section{Intuitions for Main Results} \label{sec:intuition}

We provide intuitions for Corollaries \ref{corMyopic}-\ref{corEfficient} and Theorem \ref{thm:general} together in the sections below.

\subsection{High-Level Argument}
Agents choose signals by comparing the marginal value of different observations. Thus, signal acquisitions eventually concentrate on a set $\mathcal{S}$ if and only if the marginal values of signals in that set become persistently higher than those of other signals. 

In settings of \emph{exactly sufficient information}, in which agents must observe all available signals in order to learn $\omega$, it can be shown that agents will eventually observe all signals and learn all states (including all of the confounding terms). Thus, agents will come to evaluate signals by a prior-independent ``asymptotic" valuation, which also allows them to identify the best set of signals and approximate the optimal frequency. 

When information is overabundant, agents can learn $\omega$ from many different (proper) subsets of signals, and there is no guarantee that agents will observe signals in the best set at all. This complicates the analysis, since the marginal value of any given signal depends critically on which signals have been observed previously. It is exactly this difference that leads to our learning traps result (Corollary \ref{corMyopic}): Observation of different minimal spanning sets in the long run can be sustained by prior beliefs (and resulting posterior beliefs) that overvalue the signals within the set relative to signals outside of the set. 

However, we show that our previous argument for the exactly sufficient information case holds ``subspace by subspace." That is, as agents repeatedly acquire signals from any fixed subspace of signals, they will eventually discover the asymptotic marginal values of each signal \emph{in that subspace}. In the long run, agents identify and choose from the best set of signals within that subspace. Thus, only those sets of signals that are best in their subspace are potentially ``self-sustaining." And if all sets of signals that reveal $\omega$ span the entire space, agents will identify the best set of signals overall and achieve efficient information aggregation. 

\subsection{Gradient Descent Dynamics}
To provide further intuition, we introduce the following ``normalized" asymptotic posterior variance function $V^*$ (which takes as input frequency vectors $\lambda \in \Delta^{N-1}$):
\[
V^*(\lambda) = \lim_{t\rightarrow \infty} t \cdot V(\lambda t).
\]
We establish the following relationships between $V^*$ and $V$. First, signal acquisitions chosen according to a frequency vector that minimizes $V^*$ will asymptotically also minimize the posterior variance function $V$ (Lemma \ref{lemmReductionToAsympVar}); this justifies our interest in $V^*$. Second, $V^*(\lambda)$ is convex in $\lambda$ and its unique minimum is the optimal frequency vector $\lambda^*$ (Lemma \ref{lemmUniqueMinimizer}).  So the question of whether efficient information aggregation obtains is equivalent to the question of whether the frequency vector $m(t)/t$ comes to minimize $V^*$. 

By construction, each agent chooses the signal that minimizes the value of $V$. Under a certain condition, this is also the signal that roughly achieves the greatest reduction in $V^*$. Thus, we can think of society's frequency vector $\lambda(t):= m(t)/t$ as evolving according to a gradient descent dynamic: At each (large) $t$, $\lambda(t)$ moves in a direction that minimizes $V^*$. Since $V^*$ is a convex function, we might expect gradient descent to be well-behaved, eventually converging to the global minimizer $\lambda^*$. 

While the above argument appears to be in conflict with our learning traps result, our problem has the special feature that descent can only occur along $N$ directions (indexed by the available signals).\footnote{This constraint corresponds to our assumption that each agent acquires a single, discrete, observation of a chosen signal.} This limitation is without loss whenever $V^*$ is differentiable, since all directional derivatives can then be rewritten as convex combinations of the partial derivatives along basis vectors. The function $V^*$, however, is \emph{not} differentiable everywhere. Consider our learning trap example with signals
\begin{align*}
X_1 &= \omega + \epsilon_1\\
X_2 &= 3\omega + b_1 +\epsilon_2\\
X_3 &= b_1 + \epsilon_3
\end{align*}
and set the frequency vector to be $\lambda=(1,0,0)$. It is easy to verify that beliefs are made less precise if we re-assign weight from $X_1$ to $X_2$, or from $X_1$ to $X_3$. But beliefs are made more precise if we simultaneously re-assign weight from $X_1$ to \emph{both} $X_2$ and $X_3$. This means that the derivative of $V^*$ in either direction $(-1, 1, 0)$ or $(-1, 0, 1)$ is positive, while its derivative in the direction $(-1, \frac{1}{2}, \frac{1}{2})$ is in fact negative. Hence, $V^*$ is not differentiable at $\lambda$.

Our constrained version of gradient descent can become stuck at vectors $\lambda$ such as this, so that agents repeatedly sustain the frequency vector $\lambda$ instead of moving to another frequency vector with smaller $V^*$. This is reflected in our learning traps result (Corollary \ref{corMyopic} and part (a) of Theorem \ref{thm:general}). On the other hand, it can be shown that $V^*$ is differentiable at $\lambda$ whenever $\lambda$ places nonzero weight on a spanning set of $K$ signals. This explains why efficient information aggregation occurs under the assumption of Corollary \ref{corEfficient}, where the need to recover $\omega$ requires society to learn all the states (and hence to observe a set of signals that spans all of $\mathbb{R}^K$). Repeating this argument for each subspace yields part (b) of Theorem \ref{thm:general}. 

\begin{remark} 
The above intuition connects to a literature on learning convergence in potential games \citep{MondererShapley,Sandholm}. Define an $N$-player game where each player $i$ chooses a number $\lambda_i \in \mathbb{R}_{+}$ and receives payoff $-\left(\sum_{j = 1}^{N} \lambda_j\right) \cdot V^*(\lambda) = -V^*\left(\lambda/\sum_{j = 1}^{N} \lambda_j\right)$.\footnote{Thus players' actions are normalized to a frequency vector $\lambda/\sum_{j = 1}^{N} \lambda_j$.} Then, we have a potential game with (exact) potential function $-V^*$, and our long-run observation sets correspond to equilibria of this game. In finite potential games and infinite games with a differentiable potential function, pure-strategy Nash equilibria can be related to the extreme points of the potential function. However, our game described above is an infinite potential game with a non-differentiable potential function. It is known that Nash equilibria in such games need not occur at extreme points, and this is consistent with our result on learning traps.

We note that the connection to potential games is not sufficient to derive our main results, since our agents receive payoff $-V(\lambda t)$ rather than its asymptotic variant $V^*$. This difference is nontrivial, because whether we can substitute analysis of $V^*$ for analysis of $V$ depends on the prior belief as well as the endogenous path of signal acquisitions.\footnote{Note that the prior belief\textemdash which is outside of the description of the asymptotic potential game\textemdash influences which outcome agents will converge to. This is also a feature special to our setting.} 
\end{remark}

\section{Interventions}\label{sec:intervention}
Section \ref{sec:mainResults} demonstrated the possibility for sequential information acquisition to result in inefficient learning. We ask now whether it is possible for a policymaker to influence agents towards efficient learning. Naturally, this question applies only when agents would otherwise potentially achieve a suboptimal speed of learning (with conditions given in part (a) of Theorem \ref{thm:general}). 

We compare several possible policy interventions: Increasing the \emph{quality} of information acquisition (so that each signal acquisition is more informative); restructuring incentives so that agents' payoffs are based on information obtained over several periods (equivalent to acquisition of \emph{multiple signals} each period); and providing a one-shot release of \emph{free information}, which can then guide subsequent acquisitions. These interventions have different levels of effectiveness, as we explain below. 

\subsection{More Precise Information}\label{sec:preciseInfo}
Consider first an intervention in which the precision of signal draws is increased uniformly over signals. For example, if different signals correspond to measurement of different neurochemicals in a group of lab subjects, a government agency can provide researchers with funding that permits recruitment of more subjects. This improves the quality of the estimate regardless of which neurochemical the researcher chooses to measure.

We model this intervention by supposing that each signal acquisition produces $B$ independent observations from that source (with the main model corresponding to $B=1$). Our result below, which follows from part (a) of Theorem \ref{thm:general}, says that providing more informative signals is of limited effectiveness: Any set of signals that is a potential learning trap given $B=1$ remains a potential learning trap under arbitrary improvements to signal precision.  

\begin{corollary} 
Suppose that for $B=1$, there is a set of priors given which signals in $\mathcal{S}$ are (exclusively) viewed in the long run. Then, for every $B \in \mathbb{Z}_+$, there is a set of priors given which $\mathcal{S}$ is exclusively viewed in the long run. 
\end{corollary}

However, the sets of prior beliefs corresponding to different values of $B$ need not be the same. For a \emph{fixed} prior belief, subsidizing higher quality acquisitions may or may not move the community out of a learning trap. To see this, consider first the informational environment and prior belief from Example \ref{ex:LT}. Increasing the precision of signals is ineffective here: The first agent chooses $X_1$ regardless of the value of $B$, and our previous logic again implies that each subsequent agent also chooses signal $X_1$. Thus, the set $\{X_1\}$ remains a learning trap. In Appendix \ref{appx:preciseInfo}, we provide a contrasting example in which increasing the precision of signals can indeed break agents out of a learning trap from a specified prior belief. Which of these examples is relevant depends on fine details of the informational environment as well as the prior, which the policymaker may not know in practice.

\subsection{Batches of Signals}
Another possibility is to restructure the incentive scheme so that agents' payoffs are based on information obtained over several periods, equivalent to acquisition of a batch of signals each period. For example, evaluation of researchers can be based on a set of papers, or researchers can be given labs and permitted to direct the work of multiple individuals simultaneously.

Both of these approaches for restructuring the environment can be modeled as permitting each agent to allocate $B$ observations across the sources (where $B=1$ returns the main model). Note the key difference from the previous intervention: Here it is possible for the $B$ observations to be allocated across \emph{different} signals. We show that it is possible to guarantee efficient information aggregation in this case:

\begin{proposition}\label{prop:batch}
Under Unique Minimizer, there is a $B$ such that given acquisition of $B$ signals every period, long-run frequency is $\lambda^*$ starting from every prior belief.
\end{proposition}

\noindent Thus, given sufficiently many observations each period, agents will allocate observations in a way that eventually approximates the optimal frequency.

The number of observations needed, however, depends on subtle details of the informational environment. In particular, the required $B$ cannot be uniformly bounded over all environments of fixed size (number of states $K$ and number of signals $N$).
The required $B$ instead depends on two properties: First, it depends on how well the optimal frequency $\lambda^*$ can be approximated via allocation of $B$ observations.\footnote{For example, $\lambda^*=(1/2,1/2)$ can be achieved exactly using two observations, while $\lambda^*=(3/8, 5/8)$ cannot.}
Second, it depends on the difference in learning speed between the best set and the next best minimal spanning set; this difference determines the ``slack" that is permitted in the approximation of $\lambda^*$. Thus, small batch sizes $B$ are sufficient when the optimal frequency $\lambda^*$ can be well-approximated using a small number of observations, or when there are large efficiency gains from observing the best set. See Appendix \ref{appx:batch} for further details. 

\subsection{Free Information}
Finally, we consider provision of free information to the community.  We can think of this as releasing information that a policymaker knows, or as a reduced form for funding specific kinds of research, the results of which are made public.

Formally, the policymaker chooses $M$ signals
$
\langle p_j, \theta \rangle + \mathcal{N}(0,1)
$, 
where each $\Vert p_j \Vert_2 \leq \gamma$, so that signal precisions are bounded by $\gamma^2$. At time $t=0$, independent realizations of these signals are made public. All subsequent agents update their prior beliefs based on this free information in addition to the history of signal acquisitions thus far. 

We show that given a sufficient number of (kinds of) signals, of sufficiently high precision,  efficient learning can be guaranteed. Specifically, if $k \leq K$ is the size of the optimal set $\mathcal{S}^*$, then $k-1$ precise signals are sufficient to guarantee efficient learning: 

\begin{proposition}\label{prop:freeSignal} 
Let $k:=\vert \mathcal{S}^* \vert$. Under Unique Minimizer, there exists a $\gamma<\infty$, and $k-1$ signals with $\Vert p_j \Vert_2 \leq \gamma$, such that with these free signals provided at $t = 0$, society's long-run frequency is $\lambda^*$ starting from every prior belief. 
\end{proposition}

The proof is by construction. We show that as long as agents understand those confounding terms that appear in the best set of signals (these parameters have dimension $k-1$), they will come to evaluate the signals in the best set according to their asymptotic marginal values.\footnote{This intervention requires knowledge of the full correlation structure, and also which set $\mathcal{S}^*$ is best. An alternative intervention, with higher demands on information provision but lower demands on knowledge of the environment, is to provide $K-1$ (sufficiently precise) signals about all of the confounding terms.}

This intervention is most relevant in settings in which a technological advance could greatly speed up progress, but development of the technology is slow and tedious. For example, suppose that high-resolution brain scans would allow for rapid understanding of depression, but the current imaging technology is very poor. Researchers working to understand depression may prefer to exploit existing technologies, rather than contribute to development of this new technology. The government can intervene by funding preliminary development of brain imaging, which then encourages researchers to begin using brain scans. Once low-resolution brain scans are common, the payoff to advancing the imaging technology increases, and even short-sighted researchers may contribute to this agenda. In this way, provision of free information can nudge agents onto the right path of learning.

\section{Extensions} \label{sec:extensions}

\subsection{General Payoff Functions}
All of our main results extend when each agent $t$ chooses an action to maximize an arbitrary individual payoff function $u_t(a_t,\omega)$ (recall that previously we restricted to $u_t(a_t,\omega)=-(a_t-\omega)^2$). We require only that these payoff functions are nontrivial in the following sense:
\begin{assumption}[Payoff Sensitivity to Mean]\label{assumptionSensitivity}
For every $t$, any variance $\sigma^2 > 0$ and any action $a^* \in A$, there exists a positive Lebesgue measure of $\mu$ for which $a^*$ does \emph{not} maximize $\mathbb{E}[u_t(a, \omega) \mid \omega \sim \mathcal{N}(\mu, \sigma^2)]$.
\end{assumption}
\noindent That is, for every belief variance, the expected value of $\omega$ affects the optimal action to take. This rules out cases with a ``dominant" action and ensures that each agent \emph{strictly} prefers to choose the most informative signal. Since the signal that minimizes the posterior variance about $\omega$ Blackwell-dominates every other signal,\footnote{See, e.g., \cite{HansenTorgersen}.} each agent's information acquisition remains unchanged. 

However, the interpretation of the optimal benchmark (that we defined in Section \ref{sec:optimalAsymp}) is more limited. Specifically, while the optimal frequency can still be interpreted as maximizing information revelation, the relationship to the social planner problem (Proposition \ref{prop:highDelta}) may fail. A detailed discussion is relegated to Appendix \ref{appx:genpayoff}.

\subsection{Low Altruism} 
So far we have assumed that agents care only to maximize the accuracy of their own prediction of the payoff-relevant state. Consider a generalization in which agents are slightly altruistic; that is, each agent $t$ chooses a signal as well as an action $a_t$ to maximize discounted payoffs $\mathbb{E} \left[\sum_{t' \geq t} \delta^{t' - t} \cdot u(a_t, \omega) \right]$, assuming that future agents will behave similarly. Note that $\delta = 0$ returns our main model. We show in Appendix \ref{appx:lowDelta} that for $\delta$ sufficiently small, part (a) of Theorem \ref{thm:general} continues to hold (in every equilibrium of this game). So the existence of learning traps is robust to a small degree of altruism. 

\subsection{Multiple Payoff-Relevant States} 
In our main model, only one of the $K$ persistent states is payoff-relevant. Consider a generalization in which each agent predicts (the same) $r \leq K$ unknown states and his payoff is determined via a weighted sum of quadratic losses. We show in Appendix \ref{appx:multiStates} that all of our main results extend to this setting. The possibility for agents to have payoffs that depend on \emph{heterogeneous} states is also interesting, and we leave this for future work.

\section{Conclusion}
We study a model of sequential learning, where short-lived agents choose what kind of information to acquire from a large set of available information sources. Because agents do not internalize the impact of their information acquisitions on later decision-makers, they may acquire information inefficiently (from a social perspective). Inefficiency is not guaranteed, however: Depending on the informational environment, myopic concerns can endogenously push agents to identify and observe only the most informative sources.

Our main results separate these possibilities, and reveal that the extent of \emph{learning spillovers} is essential to determining which outcome emerges. Specifically, does information about unknowns of immediate societal interest (i.e., the payoff-relevant state) also teach about unknowns that are only of indirect value (i.e., the confounding terms)?
 
When such spillovers are present, simple incentive schemes for information acquisition\textemdash in which agents care only about immediate contributions to knowledge\textemdash are sufficient for efficient long-run learning. When these spillovers are not built into the environment, other incentives are needed.  For example, forward-looking funding agencies can encourage investment in the confounding terms (our ``free information" intervention). Alternatively, agents can be evaluated on the basis of a body of work (our ``multiple signal" intervention). These observations are consistent with practices that have arisen in academic research, including the establishment of third-party funding agencies (e.g.\ the NSF) to support basic science and methodological research, and the evaluation of researchers based on advancements developed across several papers (e.g.\ tenure and various prizes).

We conclude below with brief mention of additional directions. So far we have considered the demand for information given an exogenous set of information sources. In a complementary problem to ours, information sources choose the information they provide in order to maximize demand. Our Theorem \ref{thm:totalOpt} implies the following comparative static: If signal $i$ is viewed with positive frequency in the optimal benchmark, then this frequency is (locally) decreasing in its precision. Thus, if demand is interpreted as $\lambda_i^*$ (the long-run frequency with which source $i$ is optimally viewed), sources face conflicting incentives: They want to provide information sufficiently precise to be included in the best set and receive viewership at all, but subject to this, they want to provide signals as imprecise as possible. These conflicting forces suggest that characterization of the equilibrium provisions of information precision is subtle.

Finally, while we have described our setting as choice between information sources, our model may apply more generally to choice between actions with complementarities. For example, suppose a sequence of managers take actions that have externalities for future managers, and each manager seeks to maximize performance of the company during his tenure. The concepts we have developed here of \emph{efficient information aggregation} and \emph{learning traps} have natural analogues in that setting (actions that maximize the company's long-term welfare, versus those that do not). Relative to the general setting, we study here a class of complementarities that are micro-founded in correlated signals. It is an interesting question of whether and how the forces we find here generalize to other sorts of complementarities.


\appendix

\section{Appendix}
The structure of the appendix follows that of the paper. We provide proofs for the results listed in the main text, in the order in which they appeared; the only exception is that the proof of Proposition \ref{prop:highDelta} relies on tools we develop in the other proofs, and so it is given at the end. Other results and examples are deferred to a separate Online Appendix. 

\subsection{Posterior Variance Function} \label{appx:prelim}

\subsubsection{Monotonicity and Convexity}
Here we review and extend a basic result from \citet{LiangMuSyrgkanis}. Specifically, we show that the posterior variance about $\omega$ weakly decreases over time, and the marginal value of any signal decreases in its signal count. 

\begin{lemma}\label{lemmVar}
Given prior covariance matrix $\Sigma^0$ and $q_i$ observations of each signal $i$, society's posterior variance about $\omega$ is
\begin{equation}\label{eq:f}
V(q_1, \dots, q_N) = \left[ ((\Sigma^0)^{-1} + C'QC)^{-1} \right]_{11}
\end{equation}
where $Q = \diag(q_1, \dots, q_N)$. The function $V$ is decreasing and convex in each $q_i$ whenever these arguments take non-negative real values.
\end{lemma}

\begin{proof}
Note that $(\Sigma^0)^{-1}$ is the prior precision matrix and $C'QC = \sum_{i = 1}^{N} q_i \cdot [c_i c_i']$ is the total precision from the observed signals. Thus (\ref{eq:f}) simply represents the fact that for Gaussian prior and signals, the posterior precision matrix is the sum of the prior and signal precision matrices. To prove the monotonicity of $V$, consider the partial order $\succeq$ on positive semi-definite matrices where $A \succeq B$ if and only if $A - B$ is positive semi-definite. As $q_i$ increases, the matrix $Q$ and $C'QC$ increase in this order. Thus the posterior covariance matrix $((\Sigma^0)^{-1} + C'QC)^{-1}$ decreases in this order, which implies that the posterior variance about $\omega$ decreases.

To prove that $V$ is convex, it suffices to prove that $V$ is \emph{midpoint-convex} since the function is clearly continuous.\footnote{A function $V$ is midpoint-convex if the inequality $V(a) + V(b) \geq 2 V(\frac{a+b}{2})$ always holds. Every continuous function that is midpoint-convex is also convex.} Take $q_1, \dots, q_N$, $r_1, \dots, r_N \in \mathbb{R}_{+}$ and let $s_i = \frac{q_i + r_i}{2}$. Define the corresponding diagonal matrices to be $Q$, $R$, $S$. Observe that $Q + R = 2S$. Thus by the AM-HM inequality for positive-definite matrices, we have 
\[
((\Sigma^0)^{-1} + C'QC)^{-1} + ((\Sigma^0)^{-1} + C'RC)^{-1} \succeq 2((\Sigma^0)^{-1} + C'SC)^{-1}.
\]
Using (\ref{eq:f}), we conclude that
\[
V(q_1, \dots, q_N) + V(r_1, \dots, r_N) \geq 2V(s_1, \dots, s_N). 
\]
This proves the (midpoint) convexity of $V$. 
\end{proof}

\subsubsection{Inverse of Positive Semi-definite Matrices}
For future use, we provide a definition of $[X^{-1}]_{11}$ for positive \emph{semi-definite} matrices $X$. When $X$ is positive definite, its eigenvalues are strictly positive, and its inverse matrix is defined as usual. In general, we can apply the Spectral Theorem to write 
\[
X = U D U', 
\]
where $U$ is a $K \times K$ orthogonal matrix whose columns are eigenvectors of $X$, and $D = \diag (d_1, \dots, d_K)$ is a diagonal matrix consisting of non-negative eigenvalues. Even if some of these eigenvalues are zero, we can think of $X^{-1}$ as
\[
X^{-1} = (U D U')^{-1} = U D^{-1} U' = \sum_{j = 1}^{K} \frac{1}{d_j} \cdot [u_j u_j']
\]
where $u_j$ is the $j$-th column vector of $U$. We thus define 
\begin{equation}\label{eq:matrixInverse}
[X^{-1}]_{11} := \sum_{j = 1}^{K} \frac{(\langle u_j, e_1 \rangle)^2}{d_j},
\end{equation}
with the convention that $\frac{0}{0} = 0$ and $\frac{z}{0} = \infty$ for any $z > 0$. Note that by this definition, 
\[
[X^{-1}]_{11} = \lim_{\epsilon \rightarrow 0_{+}} \left(\sum_{j = 1}^{K} \frac{(\langle u_j, e_1 \rangle)^2}{d_j + \epsilon} \right) = [(X + \epsilon I_K)^{-1}]_{11},
\]
since the matrix $X + \epsilon I_K$ has the same set of eigenvectors as $X$, with eigenvalues increased by $\epsilon$. Hence our definition of $[X^{-1}]_{11}$ is a continuous extension of the usual definition to positive semi-definite matrices.

\subsection{Proof of Theorem \ref{thm:totalOpt}} \label{appx:totalOpt}

\subsubsection{Asymptotic Posterior Variance Function}
We first approximate the posterior variance as a function of the frequencies with which each signal is observed. Specifically, for any $\lambda \in \mathbb{R}_{+}^{N}$, define
\[
V^*(\lambda) := \lim_{t \rightarrow \infty} t \cdot V(\lambda t).
\]
The following result shows $V^*$ to be well-defined and computes its value:
\begin{lemma}\label{lemmAsympVar} 
Let $\Lambda = \diag(\lambda_1, \dots, \lambda_N)$. Then 
\begin{equation}\label{eq:f*}
\begin{split}
V^*(\lambda) = [(C'\Lambda C)^{-1}]_{11}
\end{split}
\end{equation}
The value of $[(C'\Lambda C)^{-1}]_{11}$ is well-defined, see (\ref{eq:matrixInverse}). 
\end{lemma}

\begin{proof}
Recall that $V(q_1, \dots, q_N) = \left[ ((\Sigma^0)^{-1} + C'QC)^{-1} \right]_{11}$ with $Q = \diag(q_1, \dots, q_N)$. Thus
\[
t \cdot V(\lambda_1 t, \dots, \lambda_N t) = \left[ \left(\frac{1}{t}(\Sigma^0)^{-1} + C' \Lambda C \right)^{-1} \right]_{11}.
\] 
Hence the lemma follows from the continuity of $[X^{-1}]_{11}$ in the matrix $X$.
\end{proof}

We note that $C'\Lambda C$ is the Fisher Information Matrix when signals are observed according to frequencies $\lambda$. Thus the above lemma can also be seen as an application of the Bayesian Central Limit Theorem.

\subsubsection{Reduction to the Study of $V^*$}
The development of the function $V^*$ is useful for the following reason:

\begin{lemma}\label{lemmReductionToAsympVar}
Suppose $\tilde{\lambda}$ uniquely minimizes $V^*(\cdot)$ subject to $\lambda \in \Delta^{N-1}$ (the $(N-1)$-dimensional simplex). Then, the $t$-optimal divisions satisfy $n_i(t) \sim \tilde{\lambda}_i \cdot t$ for all $i$. 
\end{lemma}

\begin{proof}
Fix any increasing sequence of times $t_1, t_2, \dots$. It suffices to show that whenever the limit $\lambda_i := \lim_{m \rightarrow \infty} \frac{n_i(t_m)}{t_m}$ exists for each $i$, this limit $\lambda$ must be $\tilde{\lambda}$. Suppose not, then by assumption $V^*(\lambda) > V^*(\tilde{\lambda})$. For $\epsilon > 0$, define another vector $\hat{\lambda} \in \mathbb{R}_{+}^{N}$ with $\hat{\lambda}_i = \lambda_i + \epsilon, \forall i$. By the continuity of $V^*$, it holds that $V^*(\hat{\lambda}) > V^*(\tilde{\lambda})$ for sufficiently small $\epsilon$. 

Since $\lambda_i = \lim_{m \rightarrow \infty} \frac{n_i(t_m)}{t_m}$, there exists $M$ sufficiently large such that $n_i(t_m) \leq (\lambda_i + \epsilon) \cdot t_m$ for each $i$ and $m \geq M$. Hence, for $m \geq M$, 
\[
t_m \cdot V(n_1(t_m), \dots, n_N(t_m)) \geq t_m \cdot V(\hat{\lambda}_1 \cdot t_m, \dots, \hat{\lambda}_N \cdot t_m) \rightarrow V^*(\hat{\lambda})
\]
where the inequality uses the monotonicity of $V$. On the other hand, 
\[
t_m \cdot V(\tilde{\lambda}_1 \cdot t_m, \dots, \tilde{\lambda}_N \cdot t_m) \rightarrow V^*(\tilde{\lambda}).
\]
Comparing the above two displays, we see that for sufficiently large $m$, 
\[
V(n_1(t_m), \dots, n_K(t_m)) > V(\tilde{\lambda}_1 \cdot t_m, \dots, \tilde{\lambda}_N \cdot t_m).
\]
But this contradicts the $t$-optimality of the division $n(t_m)$, as society could do better by following frequencies $\tilde{\lambda}$. The lemma is thus proved.  
\end{proof}

\subsubsection{Crucial Lemma}
We pause to demonstrate the following technical lemma:
\begin{lemma}\label{lemmLinearCombination}
Suppose $\mathcal{S}^* = \{1, \dots, K\}$ uniquely minimizes $\phi(\cdot)$ and let $C^*$ be the $K \times K$ submatrix of $C$ corresponding to the first $K$ signals. Further suppose $\beta_j^{\mathcal{S}^*} = [(C^*)^{-1}]_{1j}$ is positive for $1 \leq j \leq K$. Then for any signal $i > K$, if we write $c_i = \sum_{j = 1}^{K} \alpha_j \cdot c_j$ (which is a unique representation), then $\lvert \sum_{j = 1}^{K} \alpha_j \rvert < 1$. 
\end{lemma}

\begin{proof}
By assumption, we have the vector identity
\[
e_1 = \sum_{j = 1}^{K} \beta_j \cdot c_j \quad \text{with} ~ \beta_j = [(C^*)^{-1}]_{1j} > 0.
\]
Suppose for contradiction that $\sum_{j = 1}^{K} \alpha_j \geq 1$ (the opposite case where the sum is $\leq -1$ can be similarly treated). Then some $\alpha_j$ must be positive. Without loss of generality, we assume $\frac{\alpha_1}{\beta_1}$ is the largest among such ratios. Then $\alpha_1 > 0$ and 
\[
e_1 = \sum_{j = 1}^{K} \beta_j \cdot c_j = \left(\sum_{j = 2}^{K} \left(\beta_j - \frac{\beta_1}{\alpha_1} \cdot \alpha_j\right) \cdot c_j \right) + \frac{\beta_1}{\alpha_1} \cdot \left(\sum_{j = 1}^{K} \alpha_j \cdot c_j \right)
\]
This represents $e_1$ as a linear combination of the vectors $c_2, \dots, c_K$ and $c_i$, with coefficients $\beta_2 - \frac{\beta_1}{\alpha_1} \cdot \alpha_2, \dots, \beta_K - \frac{\beta_1}{\alpha_1} \cdot \alpha_K$ and $\frac{\beta_1}{\alpha_1}$. Observe that these coefficients are non-negative: For each $2 \leq j \leq K$, $\beta_j - \frac{\beta_1}{\alpha_1} \cdot \alpha_j$ is clearly positive if $\alpha_j \leq 0$ (since $\beta_j > 0$). And if $\alpha_j > 0$, then by assumption $\frac{\alpha_j}{\beta_j} \leq \frac{\alpha_1}{\beta_1}$ and $\beta_j - \frac{\beta_1}{\alpha_1} \cdot \alpha_j$ is again non-negative. 

By definition, $\phi(\{2, \dots, K, i\})$ is the sum of the absolute value of these coefficients. This sum is 
\[
\sum_{j = 2}^{K} \left(\beta_j - \frac{\beta_1}{\alpha_1} \cdot \alpha_j\right) + \frac{\beta_1}{\alpha_1} = \sum_{j = 1}^{K} \beta_j + \frac{\beta_1}{\alpha_1} \cdot \left(1 - \sum_{j = 1}^{K} \alpha_j\right) \leq \sum_{j = 1}^{K} \beta_j.
\]
But then $\phi(\{2, \dots, K, i\}) \leq \phi(\{1, 2, \dots, K\})$, leading to a contradiction. Hence the lemma must be true.
\end{proof}

\subsubsection{Proof of Theorem \ref{thm:totalOpt} when $\lvert \mathcal{S}^* \rvert = K$}
Given Lemma \ref{lemmReductionToAsympVar}, Theorem \ref{thm:totalOpt} will follow once we show that $\lambda^*$ uniquely minimizes $V^*(\cdot)$ over the simplex; recall that $\lambda^*$ denotes the optimal frequencies for the minimal spanning set $\mathcal{S}^*$ that minimizes $\phi$. In this section, we prove that $\lambda^*$ is indeed the unique minimizer whenever this best subset $\mathcal{S}^*$ contains exactly $K$ signals. Later on we will prove the same result even when $\lvert \mathcal{S}^* \rvert < K$, but that proof will require additional techniques. 

\begin{lemma}\label{lemmUniqueMinimizer}
For $\lambda \in \Delta^{N-1}$, the function $V^*(\lambda)$ is uniquely minimized at $\lambda = \lambda^*$. 
\end{lemma}

\begin{proof}
First, we assume that $[(C^*)^{-1}]_{1i}$ is positive for $1 \leq i \leq K$. This is without loss because we can always work with the ``negative" of any signal (replace $c_i$ with $-c_i$), which does not affect agents' behavior. 

Since $V(q_1, \dots, q_N)$ is convex in its arguments, $V^*(\lambda) = \lim_{t \rightarrow \infty} t \cdot V(\lambda_1t, \dots, \lambda_Nt)$ is also convex in $\lambda$. To show $\lambda^*$ uniquely minimizes $V^*$, we only need to show $\lambda^*$ is a \emph{local minimum}. In other words, it suffices to show $V^*(\lambda^*) < V^*(\lambda)$ for any $\lambda$ that belongs to an $\epsilon$-neighborhood of $\lambda^*$. By assumption, $\mathcal{S}^*$ is minimally-spanning and so its signals are linearly independent. Thus its signals must span all of the $K$ states. From this it follows that the $K \times K$ matrix $C' \Lambda^* C$ is positive definite, and by (\ref{eq:f*}) the function $V^*$ is differentiable near $\lambda^*$ (see Remark \ref{remarkDifferentiability} below).

We claim that the partial derivatives of $V^*$ satisfy the following inequality: 
\begin{equation}\label{eq:partialsIneq}\tag{*}
\partial_K V^*(\lambda^*) < \partial_i V^*(\lambda^*) \leq 0, \forall i > K.
\end{equation}
Once this is proved, we will have, for $\lambda$ close to $\lambda^*$, 
\begin{equation}\label{eq:f*Ineq}
V^*(\lambda_1, \dots, \lambda_K, \lambda_{K+1}, \dots, \lambda_N) \geq V^*\left(\lambda_1, \dots, \lambda_{K-1}, \sum_{k=K}^N \lambda_k, 0, \dots, 0\right) \geq V^*(\lambda^*). 
\end{equation}
The first inequality is based on (\ref{eq:partialsIneq}) and differentiability of $V^*$, while the second inequality is because $\lambda^*$ uniquely minimizes $V^*$ when only the first $K$ signals are observed. Moreover, when $\lambda \neq \lambda^*$, one of these inequalities is strict so that $V^*(\lambda) > V^*(\lambda^*)$ holds strictly.

\bigskip

To prove (\ref{eq:partialsIneq}), we recall that
\[
V^*(\lambda) = e_1' (C' \Lambda C)^{-1}e_1. 
\]
Since $\Lambda = \diag(\lambda_1, \dots, \lambda_N)$, its derivative is $\partial_i \Lambda = \Delta_{ii}$, which is an $N \times N$ matrix whose $(i,i)$-th entry is $1$ with all other entries equal to zero. Using properties of matrix derivatives, we obtain
\[
\partial_i V^*(\lambda) = - e_1' (C' \Lambda C)^{-1} C' \Delta_{ii} C (C' \Lambda C)^{-1} e_1. 
\]
As the $i$-th row vector of $C$ is $c_i'$, $C' \Delta_{ii} C$ is the $K \times K$ matrix $c_i c_i'$. The above simplifies to 
\[
\partial_i V^*(\lambda) = - [e_1' (C' \Lambda C)^{-1} c_i]^2.
\]
At $\lambda = \lambda^*$, the matrix $C' \Lambda C$ further simplifies to $(C^*)' \cdot \diag(\lambda^*_1, \dots, \lambda^*_K) \cdot (C^*)$, which is a product of $K \times K$ invertible matrices. We thus deduce that 
\[
\partial_i V^*(\lambda^*) = - \left[e_1' \cdot (C^*)^{-1} \cdot \diag\left(\frac{1}{\lambda^*_1}, \dots, \frac{1}{\lambda^*_K}\right) \cdot ((C^*)')^{-1} \cdot c_i \right]^2.
\]
Crucially, note that the term in the brackets is a linear function of $c_i$. To ease notation, we write $v' = e_1' \cdot (C^*)^{-1} \cdot \diag\left(\frac{1}{\lambda^*_1}, \dots, \frac{1}{\lambda^*_K}\right) \cdot ((C^*)')^{-1}$ and $\gamma_i = \langle v, c_i \rangle$. Then 
\begin{equation}\label{eq1}
\partial_i V^*(\lambda^*) = - \gamma_i^2, ~1 \leq i \leq N.
\end{equation}

For $1 \leq i \leq K$, $((C^*)')^{-1} \cdot c_i$ is just $e_i$. Thus, using the assumption $[(C^*)^{-1}]_{1j} > 0, \forall j$, we have 
\begin{equation}\label{eq2}
\gamma_i = e_1' \cdot (C^*)^{-1} \, \cdot \, \diag\left(\frac{1}{\lambda^*_1}, \dots, \frac{1}{\lambda^*_K}\right) \cdot e_i = \frac{[(C^*)^{-1}]_{1i}}{\lambda^*_i} = \sum_{j = 1}^{K} [(C^*)^{-1}]_{1j} = \phi(\mathcal{S}^*)
\end{equation}
On the other hand, choosing any $i > K$, we can uniquely write the vector $c_i$ as a linear combination of $c_1, \dots, c_K$. By Lemma \ref{lemmLinearCombination}, for any $i > K$ we have 
\begin{equation}\label{eq3}
\gamma_i = \langle v, c_i \rangle = \sum_{j = 1}^{K} \alpha_j \cdot \langle v, c_j \rangle = \sum_{j=1}^{K} \alpha_j \cdot \gamma_j = \phi(\mathcal{S}^*) \cdot \sum_{j = 1}^{K} \alpha_j,
\end{equation}
where the last equality uses (\ref{eq2}). Since $\lvert \sum_{j = 1}^{K} \alpha_j \rvert < 1$, the absolute value of $\gamma_i$ for any $i > K$ is strictly smaller than the absolute value of $\gamma_K$. This together with (\ref{eq1}) proves the desired inequality (\ref{eq:partialsIneq}), and Lemma \ref{lemmUniqueMinimizer} follows. 
\end{proof}

\begin{remark}\label{remarkDifferentiability}
The essence of this proof is the following nontrivial property: The subset $\{1, \dots, K\}$ uniquely minimizes $\phi$ among all subsets of size $K$ if and only if
\[
\phi(\{1, \dots, K\}) < \phi(\{1, \dots, K\} \cup \{i\} \backslash \{j\}), ~ \forall 1 \leq j \leq K < i \leq N.
\]
That is, if a set of $K$ signals does not minimize $\phi$, then we can improve the speed of learning by adding \emph{one} signal to replace \emph{one} existing signal. This property enables us to reduce the general problem with $N$ signals to a much simpler problem with $K+1$ signals. We are then able to use calculus to resolve the latter problem, see (\ref{eq:partialsIneq}). 

This argument breaks down if we start with a set of less than $K$ signals; see Section \ref{sec:intuition} of the main text for an example in which the reduction above is not possible. In that example, even though the partial derivatives satisfy (\ref{eq:partialsIneq}), it does not hold that every \emph{directional} derivative similarly satisfies (\ref{eq:partialsIneq}). Thus $V^*$ can fail to be differentiable at the frequency vector of interest. It is for this reason that we need a different proof of Lemma \ref{lemmUniqueMinimizer} when $\lvert \mathcal{S}^* \rvert < K$, which we present next.
\end{remark}

\subsubsection{A Perturbation Argument}
To handle the case of $\lvert \mathcal{S}^* \rvert < K$, we first extend the definition of $\phi(\cdot)$ to arbitrary sets of siganls (not necessarily minimally-spanning) as follows. For any set $\mathcal{A}$ that contains a minimal spanning set, define $\phi(\mathcal{A}) = \min_{\mathcal{S} \subset \mathcal{A}} \phi(\mathcal{S})$, where the minimum is taken over all minimal spanning sets $\mathcal{S}$ contained in $\mathcal{A}$. If such $\mathcal{S}$ does not exist (i.e., $\mathcal{A}$ is not itself spanning), we let $\phi(\mathcal{A}) = \infty$. In particular, 
\[
\phi([N]) = \min_{\mathcal{S} \subset [N]} \phi(\mathcal{S})
\]
represents the minimum asymptotic standard deviation achievable by only observing the signals in \emph{some} minimal spanning set.

Our previous analysis shows that whenever $\phi(\mathcal{S})$ is uniquely minimized by a set $\mathcal{S}$ containing exactly $K$ signals, 
\[
\min_{\lambda \in \Delta^{N-1}} V^*(\lambda) = V^*(\lambda^*) = \min_{\mathcal{S} \subset [N]} \phi(\mathcal{S})^2 = \phi([N])^2
\]
We now show this equality holds more generally.

\begin{lemma}\label{lemmAsd} 
For any coefficient matrix $C$,
\begin{equation}\label{eq:f*=Asd}
\min_{\lambda \in \Delta^{N-1}} V^*(\lambda) = \phi([N])^2.
\end{equation}
\end{lemma}

\begin{proof}
Because society can choose to focus on any minimal spanning set, it is clear that $\min_{\lambda} V^*(\lambda) \leq \phi([N])^2 = \min_{\mathcal{S}} (\phi(\mathcal{S}))^2$. It remains to prove $V^*(\lambda) \geq \phi([N])^2$ for any $\lambda \in \Delta^{N-1}$. By Lemma \ref{lemmAsympVar}, we need to show $[(C' \Lambda C)^{-1}]_{11} \geq \phi([N])^2$. 

This was already proved for \emph{generic} coefficient matrices $C$ (specifically, those for which $\phi(\mathcal{S})$ is minimized by a set of $K$ signals). But even if $C$ is ``non-generic", we can approximate it by a sequence of ``generic" matrices $C_m$.\footnote{First, we may add repetitive signals to ensure $N \geq K$. This does not affect the value of $\min V^*(\lambda)$ or $\phi([N])$. Whenever $N \geq K$, it is generically true that every minimal spanning set contains exactly $K$ signals. Moreover, the equality $\phi(\mathcal{S}) = \phi(\tilde{\mathcal{S}})$ for $\mathcal{S} \neq \tilde{\mathcal{S}}$ induces a non-trivial polynomial equation over the entries in $C$. This means we can always find $C^{(m)}$ close to $C$ such that for each coefficient matrix $C^{(m)}$, different subsets $\mathcal{S}$ (of size $K$) attain different values of $\phi$.} Along this sequence, we have 
\[
[(C_m' \Lambda C_m)^{-1}]_{11} \geq \phi_m([N])^2
\]
where $\phi_m$ is the speed of learning from the $N$ signals given by coefficient matrix $C_m$. As $m \rightarrow \infty$, the LHS above approaches $[(C' \Lambda C)^{-1}]_{11}$. Thus the lemma will follow once we show that $\limsup_{m \rightarrow \infty} \phi_m([N]) \geq \phi([N])$. 

For this we invoke the following characterization 
\[
\phi([N]) = \min_{\beta \in \mathbb{R}^{N}} \sum_{i = 1}^{N} \lvert \beta_i \rvert ~ \text{ s.t. } ~ e_1 = \sum_{i = 1}^{N} \beta_i \cdot c_i. 
\]
If $e_1 =  \sum_{i} \beta_i^{(m)} \cdot c_i^{(m)}$ along the convergent sequence, then $e_1 = \sum_{i} \beta_i \cdot c_i$ for any limit point $\beta$ of $\beta^{(m)}$. This enables us to conclude $\liminf_{m \rightarrow \infty} \phi_m([N]) \geq \phi([N])$, which is more than what we need. 
\end{proof}

\subsubsection{Proof of Theorem \ref{thm:totalOpt} when $\lvert \mathcal{S}^* \rvert < K$}
Let $\mathcal{S}^* = \{1, \dots, k\}$. We will now show that even if $k < K$, $\lambda^*$ is still the unique minimizer of $V^*(\cdot)$. This will imply Theorem \ref{thm:totalOpt} via Lemma \ref{lemmReductionToAsympVar}. Since $V^*(\lambda^*) = \phi(\mathcal{S}^*)^2 = \phi([N])^2$ by definition, we know from Lemma \ref{lemmAsd} that $\lambda^*$ does minimize $V^*$. It remains to show that $\lambda^*$ is the \emph{unique} minimizer. 

To do this, we will consider a perturbed informational environment in which signals $k+1, \dots, N$ are made slightly more precise. Specifically, let $\eta > 0$ be a small positive number. Consider an alternative signal coefficient matrix $\tilde{C}$ with $\tilde{c_i} = c_i$ for $i \leq k$ and $\tilde{c_i} = (1+\eta) c_i$ for $i > k$. Let $\tilde{\phi}(\mathcal{S})$ represent the speed of learning function for this alternative problem. Then, it is clear that $\tilde{\phi}(\mathcal{S}^*) = \phi(\mathcal{S}^*)$, while $\tilde{\phi}(\mathcal{S})$ is slightly smaller than $\phi(\mathcal{S})$ for $\mathcal{S} \neq \mathcal{S}^*$. Thus with sufficiently small $\eta$, the subset $\mathcal{S}^*$ remains the best set in this perturbed environment, and $\lambda^*$ remains the optimal frequency vector. 

Let $\tilde{V}^*$ be the asymptotic posterior variance function here, then our previous analysis shows that $\tilde{V}^*$ has minimum value $\phi(\mathcal{S}^*)^2$ on the simplex. Taking advantage of the connection between $V^*$ and $\tilde{V}^*$, we thus have
\begin{equation*}
\begin{split}
V^*(\lambda_1, \dots, \lambda_N) &= \tilde{V}^*\left(\lambda_1, \dots, \lambda_k, \frac{\lambda_{k+1}}{(1+\eta)^2}, \dots, \frac{\lambda_{N}}{(1+\eta)^2}\right) \\ 
&\geq \frac{\phi(\mathcal{S}^*)^2}{\sum_{i \leq k} \lambda_i + \frac{1}{(1+\eta)^2}\sum_{i > k} \lambda_i}.
\end{split}
\end{equation*}
The equality uses (\ref{eq:f*}) and $C' \Lambda C = \sum_{i} \lambda_i c_i c_i' = \sum_{i \leq k} \lambda_i c_i c_i' + \sum_{i > k} \frac{\lambda_i}{(1+\eta)^2} \tilde{c_i} \tilde{c_i}'$. The inequality follows from the homogeneity of $\tilde{V}^*$. 

The above display implies 
\begin{equation}\label{eq:V*lowerBound}
\forall \lambda \in \Delta^{N-1}, \quad V^*(\lambda) \geq \frac{\phi(\mathcal{S}^*)^2}{1 - \frac{2\eta + \eta^2}{(1+\eta)^2} \sum_{i > k} \lambda_i} \geq \frac{\phi(\mathcal{S}^*)^2}{1 - \eta \sum_{i > k} \lambda_i} \quad \text{for some} ~\eta > 0.
\end{equation}
Hence $V^*(\lambda) > \phi(\mathcal{S}^*)^2 = V^*(\lambda^*)$ whenever $\lambda$ puts positive weight outside of the best set. But we already know that $V^*(\lambda)$ is uniquely minimized at $\lambda^*$ when $\lambda$ is restricted to the best set. Hence $\lambda^*$ is the unique minimizer of $V^*$ over the whole simplex. This completes the proof of Theorem \ref{thm:totalOpt}. 

\subsubsection{Strengthening of Theorem \ref{thm:totalOpt}} \label{appx:n(t)finite}
In this section we show $n_i(t) = \lambda^*_i \cdot t + O(1), \forall i$, which improves upon the conclusion of Theorem \ref{thm:totalOpt}. Note that for any $(q_1, \dots, q_N)$ to be a $t$-optimal division, it is necessary that $tV(q_1, \dots, q_N) \leq tV(\lambda^* t)$. A straightforward refinement of Lemma \ref{lemmAsympVar} gives that whenever $V^*(\lambda)$ is finite, $t \cdot V(\lambda t)$ approaches $V^*(\lambda)$ at the rate of $\frac{1}{t}$. Thus we must have 
\[
V^*\left(\frac{q_1}{t}, \dots, \frac{q_N}{t}\right) \leq V^*(\lambda^*) + O\left(\frac{1}{t}\right). 
\]
Together with (\ref{eq:V*lowerBound}), this implies $\sum_{i > k} \frac{q_i}{t} = O(\frac{1}{t})$, so that signals outside of the best set are only observed finitely many times. Conditional on \emph{these} signal counts, Proposition \ref{prop:N=K} shows that society's optimal allocation over the first $k$ signals must satisfy $n_i(t) = \lambda^*_i\cdot t + O(1), \forall 1 \leq i \leq k$. This proves what we want.

\subsection{Proof of Theorem \ref{thm:general} Part (a)}\label{appx:myopicExistence}
Let signals $1, \dots, k$ (with $k \leq K$) be a minimally spanning set that is optimal in its subspace. We will demonstrate an open set of prior beliefs given which \emph{all agents} observe these $k$ signals. Since these signals are minimally spanning, they must be linearly independent. Thus we can consider linearly transformed states $\tilde{\theta}_1, \dots, \tilde{\theta}_K$ such that these $k$ signals are simply $\tilde{\theta}_1, \dots, \tilde{\theta}_k$ plus standard Gaussian noise. This linear transformation is invertible, so any prior over the original states is bijectively mapped to a prior over the transformed states. Thus it is without loss to work with the transformed model and look for prior beliefs over the transformed states. 

The payoff-relevant state $\omega$ becomes a linear combination $\lambda^*_1 \tilde{\theta}_1 + \dots + \lambda^*_k \tilde{\theta}_k$ (after scaling). Since the first $k$ signals are optimal in their subspace, Lemma \ref{lemmLinearCombination} implies that any other signal that belongs to this subspace can be written as
\[
\sum_{i = 1}^{k} \alpha_i \tilde{\theta}_i \ + \ \mathcal{N}(0,1)
\]
with $\lvert \sum_{i = 1}^{k} \alpha_i \rvert < 1$. On the other hand, if a signal does not belong to this subspace, it must take the form of 
\[
\sum_{i = 1}^{K} \beta_i \tilde{\theta}_i \ + \ \mathcal{N}(0,1)
\]
with $\beta_{k+1}, \dots, \beta_{K}$ not all equal to zero. 

\bigskip

Now consider a prior belief such that $\tilde{\theta}_1, \dots, \tilde{\theta}_K$ are \emph{independent} from each other. Given prior variances $v_1, \dots, v_K$, the reduction in the variance of $\lambda^*_1 \tilde{\theta}_1 + \dots + \lambda^*_k \tilde{\theta}_k$ by any signal $\sum_{i = 1}^{k} \alpha_i \tilde{\theta}_i + \mathcal{N}(0,1)$ is 
\[
\frac{(\sum_{i = 1}^{k} \alpha_i \lambda^*_i v_i)^2}{1+\sum_{i = 1}^{k} \alpha_i^2 v_i}
\]
If $v_1, \dots, v_k$ are small positive numbers and \emph{if the product $\lambda^*_iv_i$ is approximately constant across $1 \leq i \leq k$}, then the above is approximately $(\sum_{i = 1}^{k} \alpha_i)^2 (\lambda^*_1 v_1)^2$. Since $\lvert \sum_{i = 1}^{k} \alpha_i \rvert < 1$, we deduce that any other signal belonging to the subspace of the first $k$ signals is worse than signal $1$ (in the first period), whose variance reduction is $\frac{(\lambda^*_1 v_1)^2}{1+v_1}$. 

Meanwhile, take any signal that does not belong to the subspace. The variance reduction by such a signal $\sum_{i = 1}^{K} \beta_i \tilde{\theta}_i + \mathcal{N}(0,1)$ is 
\[
\frac{(\sum_{i = 1}^{k} \beta_i \lambda^*_i v_i)^2}{1+\sum_{i = 1}^{K} \beta_i^2 v_i}
\]
As $\beta_{k+1}, \dots, \beta_{K}$ are not all zero, the denominator above is arbitrarily large if $v_{k+1}, \dots, v_{K}$ are chosen to be large. Then, this signal is again worse than signal $1$ for the first agent, similar to the situation in Example \ref{ex:LT}.

\bigskip

To summarize, we have shown that whenever the prior variances $v_1, \dots, v_K$ satisfy the following three conditions, the first agent chooses among the first $k$ signals:
\begin{enumerate}
\item $v_1, \dots, v_k$ are close to $0$; 

\item $\lambda^*_1v_1, \dots, \lambda^*_kv_k$ have pairwise ratios close to $1$;

\item $v_{k+1}, \dots, v_{K}$ are large.\footnote{Formally, we require that for some fixed constant $\epsilon > 0$, it holds that $v_1, \dots, v_k < \epsilon$; $\max_{1 \leq i \leq k} \lambda^*_i v_i \leq (1+\epsilon) \cdot \min_{1 \leq i \leq k} \lambda^*_i v_i$; and $v_{k+1}, \dots, v_{K} > \frac{1}{\epsilon}$.} 
\end{enumerate}

To show \emph{every agent} chooses among the first $k$ signals, it suffices to check that starting from any prior belief satisfying the above conditions, the posterior beliefs after observing a signal continue to satisfy these conditions. Since variances decrease over time, the first condition is obviously satisfied. By independence, learning about $\tilde{\theta}_1, \dots, \tilde{\theta}_k$ does not affect the variances of the remaining states. So $v_{k+1}, \dots, v_{K}$ are unchanged, and the third condition is verified. Finally, the second condition holds by Proposition \ref{prop:N=K}: Since each signal $i \leq k$ is sampled according to $\lambda^*_i$, the variance $v_i$ after $t$ periods is approximately $\frac{1}{\lambda^*_i t}$. Hence part (a) of Theorem \ref{thm:general} is proved.\footnote{Strictly speaking, the above construction does not provide an \emph{open set} of prior beliefs given which agents always observe the first $k$ signals. This is because we restricted attention to priors that are independent over $\tilde{\theta}_1, \dots, \tilde{\theta}_K$. But it can be shown that the argument extends to mild prior correlation across the transformed states. We omit the somewhat cumbersome details, which do not add further intuition.}

\subsection{Proof of Theorem \ref{thm:general} Part (b)}\label{appx:myopicConvergence}

\subsubsection{Preliminary Steps}
Given any prior, let $\mathcal{A} \subset [N]$ be the set of all signals that are observed by infinitely many agents. We first show that $\mathcal{A}$ is a spanning set. 

Indeed, by definition we can find some period $t$ after which agents only observe signals in $\mathcal{A}$. Note that the variance reduction of any signal approaches zero as its signal count gets large. Thus, along society's signal path, the variance reduction is close to zero at sufficiently late periods. 

If $\mathcal{A}$ is not spanning, society's posterior variance remains bounded away from zero. Thus in the limit where each signal in $\mathcal{A}$ has infinite signal counts, there still exists some signal $j$ outside of $\mathcal{A}$ whose variance reduction is strictly positive.\footnote{To see this, let $s_1, \dots, s_N$ denote the limit signal counts, where $s_i = \infty$ if and only if $i \in \mathcal{A}$. We need to find some signal $j$ such that $V(s_j+1, s_{-j}) < V(s_j, s_{-j})$. If such a signal does not exist, then all partial derivatives of $V$ at $s$ are zero. Since $V$ is differentiable, this would imply that all directional derivatives of $V$ are also zero. By the convexity of $V$, $V$ must be minimized at $s$. However, the minimum value of $V$ is zero because there exists a spanning set. This contradicts $V(s) > 0$.} By continuity, we deduce that at any sufficiently late period, observing signal $j$ is better than observing any signal in $\mathcal{A}$. This contradicts our assumption that later agents only observe signals in $\mathcal{A}$. 

\bigskip

Now that $\mathcal{A}$ is spanning, we can take $\mathcal{S}$ to be the optimal minimal spanning set in the subspace spanned by $\mathcal{A}$. To prove Theorem \ref{thm:general} part (b), we will show that long-run frequencies are positive precisely for the signals in $\mathcal{S}$. By ignoring the initial periods, we assume without loss that \emph{only signals in $\overline{\mathcal{A}}$ are available}. It suffices to show that whenever the signals observed infinitely often span a subspace, agents eventually sample from the optimal subset $\mathcal{S}$ in that subspace. To ease notation, we assume this subspace is the entire $\mathbb{R}^{K}$, and prove the following result: 

\bigskip

\noindent \textbf{Theorem \ref{thm:general} part (b) Restated.} \emph{Suppose that the signals observed infinitely often span $\mathbb{R}^{K}$. Then society eventually observes signals in $\mathcal{S}^*$ with frequencies $\lambda^*$.}

\bigskip

The next sections are devoted to the proof of this restatement. 

\subsubsection{Estimates of Derivatives}
We introduce a few technical lemmata: 
\begin{lemma}\label{lemmSecondDerSmall}
For any $q_1, \dots, q_N$, we have 
\[
\left\lvert \frac{\partial_{jj} V(q_1, \dots, q_N)}{\partial_{j} V(q_1, \dots, q_N)} \right\rvert \leq \frac{2}{q_j}.
\]
\end{lemma}

\begin{proof}
Recall that $V(q_1, \dots, q_N) =  e_1' \cdot [(\Sigma^0)^{-1} + C'QC]^{-1} \cdot e_1$. Thus
\[
\partial_j V = -e_1' \cdot [(\Sigma^0)^{-1} + C'QC]^{-1} \cdot c_j \cdot c_j' \cdot [(\Sigma^0)^{-1} + C'QC]^{-1} \cdot e_1,
\]
and 
\[
\partial_{jj} V = 2 e_1' \cdot [(\Sigma^0)^{-1} + C'QC]^{-1} \cdot c_j \cdot c_j' \cdot [(\Sigma^0)^{-1} + C'QC]^{-1} \cdot c_j \cdot c_j' \cdot [(\Sigma^0)^{-1} + C'QC]^{-1} \cdot e_1.
\]

Let $\gamma_j = e_1' \cdot [(\Sigma^0)^{-1} + C'QC]^{-1} \cdot c_j$, which is a number. Then the above becomes
\[
\partial_j f = - \gamma_j^2; \quad \quad \partial_{jj} f = 2 \gamma_j^2 \cdot c_j' \cdot [(\Sigma^0)^{-1} + C'QC]^{-1} \cdot c_j. 
\]
Note that $(\Sigma^0)^{-1} + C'QC \succeq q_j \cdot c_j c_j'$ in matrix norm. Thus the number $c_j' \cdot [(\Sigma^0)^{-1} + C'QC]^{-1} \cdot c_j$ is bounded above by $\frac{1}{q_j}$.\footnote{Formally, we need to show that for any $\epsilon > 0$, the number $c_j' [c_j c_j' + \epsilon I_K]^{-1} c_j$ is at most $1$. Using the identify $Trace(AB) = Trace(BA)$, we can rewrite this number as 
\[
Trace([c_j c_j' + \epsilon I_K]^{-1} c_jc_j') = Trace(I_K - [c_j c_j' + \epsilon I_K]^{-1} \epsilon I_K) = K - \epsilon \cdot Trace([c_j c_j' + \epsilon I_K]^{-1}). 
\] 
The matrix $c_j c_j'$ has rank $1$, so $K-1$ of its eigenvalues are zero. Thus the matrix $[c_j c_j' + \epsilon I_K]^{-1}$ has eigenvalue $1/\epsilon$ with multiplicity $K-1$, and the remaining eigenvalue is positive. This implies $\epsilon \cdot Trace([c_j c_j' + \epsilon I_K]^{-1}) > K-1$, and then the above display yields $c_j' \cdot [(\Sigma^0)^{-1} + C'QC]^{-1} \cdot c_j < 1$ as desired. 
} This proves the lemma.
\end{proof}

Since the second derivative is small compared to the first derivative, we deduce that the variance reduction of any \emph{discrete} signal can be approximated by the partial derivative of $f$. This property is summarized in the following lemma: 

\begin{lemma}\label{lemmDiscreteDer}
For any $q_1, \dots, q_N$, we have\footnote{Note that the convexity of $V$ gives $V(q) - V(q_j + 1, q_{-j}) \leq \lvert \partial_j V(q) \rvert$. This lemma provides a converse that we need for the subsequent analysis.}
\[
V(q) - V(q_j + 1, q_{-j}) \geq \frac{q_j}{q_j+1} \lvert \partial_j V(q) \rvert.
\]
\end{lemma}

\begin{proof}
We will show the more general result: 
\[
V(q) - V(q_j + x, q_{-j}) \geq \frac{q_j x}{q_j + x} \cdot \lvert \partial_j V(q) \rvert, \forall x \geq 0.
\] 
This clearly holds at $x = 0$. Differentiating with respect to $x$, we only need to show
\[
- \partial_j V(q_j + x, q_{-j}) \geq \frac{q_j^2}{(q_j+x)^2} \lvert \partial_j V(q) \rvert, \forall x \geq 0.
\]
Equivalently, we need to show 
\[
- (q_j+x)^2 \cdot \partial_j V(q_j + x, q_{-j}) \geq -q_j^2 \cdot \partial_j V(q), \forall x \geq 0.
\]
Again, this inequality holds at $x = 0$. Differentiating with respect to $x$, it becomes 
\[
- 2(q_j+x) \cdot \partial_j V(q_j + x, q_{-j}) - (q_j+x)^2 \cdot \partial_{jj} V(q_j + x, q_{-j}) \geq 0. 
\]
This is exactly the result of Lemma \ref{lemmSecondDerSmall}. 
\end{proof}

\subsubsection{Lower Bound on Variance Reduction}
Our next result gives a lower bound on the directional derivative of $V$ along the ``optimal" direction $\lambda^*$:
\begin{lemma}\label{lemmDirectionalDer}
For any $q_1, \dots, q_N$, we have
\[
\lvert \partial_{\lambda^*} V(q) \rvert \geq \frac{V(q)^2}{\phi(\mathcal{S}^*)^2}. 
\]
\end{lemma}
 
\begin{proof}
To compute this directional derivative, we think of agents acquiring signals in fractional amounts, where a fraction of a signal is just the same signal with precision multiplied by that fraction. Consider an agent who draws $\lambda^*_i$ realizations of each signal $i$. Then he essentially obtains the following signals:
\[
Y_i = \langle c_i, \theta \rangle + \mathcal{N}\left(0, \frac{1}{\lambda^*_i}\right), \forall i. 
\]
This is equivalent to 
\[
\lambda^*_i Y_i = \langle \lambda^*_i c_i, \theta \rangle + \mathcal{N}(0, \lambda^*_i), \forall i. 
\]
Such an agent receives at least as much information as the sum of these signals:
\[
\sum_{i} \lambda^*_i Y_i = \sum_{i} \langle \lambda^*_i c_i, \theta \rangle + \sum_{i} \mathcal{N}(0, \lambda^*_i) = \frac{\omega}{\phi(\mathcal{S}^*)} + \mathcal{N}(0,1).
\]
Hence the agent's posterior precision about $\omega$ (which is the inverse of his posterior variance $V$) must increase by at least $\frac{1}{\phi(\mathcal{S}^*)^2}$ along the direction $\lambda^*$. The chain rule of differentiation yields the lemma. 
\end{proof}

We can now bound the variance reduction at late periods: 
\begin{lemma}\label{lemmVarReduction}
Fix any $q_1, \dots, q_N$. Suppose $M$ is a positive real number such that $(\Sigma^0)^{-1} + C'QC \succeq M c_j c_j'$ holds for each signal $j \in \mathcal{S}^*$. Then we have
\[
\min_{j \in \mathcal{S}^*} V(q_j + 1, q_{-j}) \leq V(q) - \frac{M}{M+1} \cdot \frac{V(q)^2}{\phi(\mathcal{S}^*)^2}. 
\]
\end{lemma}

\begin{proof}
Fix any signal $j \in \mathcal{S}^*$. Using the condition $(\Sigma^0)^{-1} + C'QC \succeq M c_j c_j'$, we can deduce the following variant of Lemma \ref{lemmDiscreteDer}:\footnote{Even though we are not guaranteed $q_j \geq M$, we can modify the prior and signal counts such that the precision matrix $(\Sigma^0)^{-1} + C'QC$ is unchanged, and signal $j$ has been observed at least $M$ times. This is possible thanks to the condition $(\Sigma^0)^{-1} + C'QC \succeq M c_j c_j'$. Then, applying Lemma \ref{lemmDiscreteDer} to this modified problem yields the result here.}
\[
V(q) - V(q_j + 1, q_{-j}) \geq \frac{M}{M+1} \lvert \partial_j V(q) \rvert.
\]
Since $V$ is always differentiable, $\partial_{\lambda^*} V(q)$ is a convex combination of the partial derivatives of $V$.\footnote{While this may be a surprising contrast with $V^*$, the difference arises because the formula for $V$ always involves a full-rank prior covariance matrix, whereas its asymptotic variant $V^*$ corresponds to a flat prior.} Thus
\[
\max_{j \in \mathcal{S^*}} ~\lvert \partial_j V(q) \rvert \geq \lvert \partial_{\lambda^*} V(q) \rvert
\]
These inequalities together with Lemma \ref{lemmDirectionalDer} complete the proof. 
\end{proof}

\subsubsection{Proof of the Restatement of Theorem \ref{thm:general} Part (b)}
We will show $V(m(t)) \sim \frac{\phi(\mathcal{S}^*)^2}{t}$, so that society eventually approximates the optimal speed of learning. Since $\lambda^*$ is the unique minimizer of $V^*$, this will imply the desired conclusion $m(t)/t \rightarrow \lambda^*$. 

To estimate $V(m(t))$, we note that for any fixed $M$, society's acquisitions $m(t)$ eventually satisfy the condition  $(\Sigma^0)^{-1} + C'QC \succeq M c_j c_j'$. This is due to our assumption that the signals observed infinitely often span $\mathbb{R}^K$, which implies that $C'QC$ becomes arbitrarily large in matrix norm. Hence, we can apply Lemma \ref{lemmVarReduction} to find that
\[
V(m(t+1)) \leq V(m(t)) - \frac{M}{M+1} \cdot \frac{V(m(t))^2}{\phi(\mathcal{S}^*)^2}
\]
for all $t \geq t_0$, where $t_0$ depends only on $M$. 

We introduce the auxiliary function $g(t) = V(m(t)) \cdot \frac{M}{(M+1) \phi(\mathcal{S}^*)^2}$. Then the above simplifies to 
\[
g(t+1) \leq g(t) - g(t)^2. 
\]
Inverting both sides, we have
\begin{equation}\label{eq:precisionIncrease}
\frac{1}{g(t+1)} \geq \frac{1}{g(t)(1-g(t))} = \frac{1}{g(t)} + \frac{1}{1-g(t)} \geq \frac{1}{g(t)} + 1.
\end{equation}
This holds for all $t \geq t_0$. Thus by induction, $\frac{1}{g(t)} \geq t - t_0$ and so $g(t) \leq \frac{1}{t-t_0}$. Going back to the posterior variance function $V$, this implies 
\begin{equation}\label{eq:posteriorVarBound}
V(m(t)) \leq \frac{M+1}{M} \cdot \frac{\phi(\mathcal{S}^*)^2}{t-t_0}.
\end{equation}
Hence, by choosing $M$ sufficiently large in the first place and then considering large $t$, we find that society's speed of learning is arbitrarily close to the optimal speed $\phi(\mathcal{S}^*)^2$. This completes the proof.

\bigskip

We comment that the above argument leaves open the possibility that some signals outside of $\mathcal{S}^*$ are observed \emph{infinitely often}, yet with \emph{zero long-run frequency}. Similar to Appendix \ref{appx:n(t)finite}, we can show this does not happen. The proof is involved, and we defer it to Appendix \ref{appx:m(t)finite}.

\subsection{Proof of Proposition \ref{prop:batch}}\label{appx:batch}
Given any history of observations, an agent can always allocate his $B$ observations as follows: He draws $\lfloor B \cdot \lambda^*_i \rfloor$ realizations of each signal $i$, and samples arbitrarily if there is any capacity remaining. Here $\lfloor ~ \rfloor$ denotes the floor function. 

Fix any $\epsilon > 0$. If $B$ is sufficiently large, then the above strategy acquires at least $(1 - \epsilon) \cdot B \cdot \lambda^*_i$ observations of each signal $i$. Adapting the proof of Lemma \ref{lemmDirectionalDer}, we see that the agent's posterior precision about $\omega$ must increase by $\frac{(1-\epsilon) B}{\phi(\mathcal{S}^*)^2}$ under this strategy. Thus the same must hold for his optimal strategy, so that society's posterior precision at time $t$ is at least $\frac{ (1-\epsilon) B t}{\phi(\mathcal{S}^*)^2}$. This implies that society's speed of learning (per signal) is at least $\frac{\phi(\mathcal{S}^*)^2}{ (1-\epsilon)}$, which can be arbitrarily close to the optimal speed $\phi(\mathcal{S}^*)^2$ with appropriate choice of $\epsilon$.

Since $\lambda^*$ is the unique minimizer of $V^*$, society's long-run frequencies must be close to $\lambda^*$. In particular, with $\epsilon$ sufficiently small, we can ensure that each signal in $\mathcal{S}^*$ are observed with positive frequencies. The restated Theorem \ref{thm:general} part (b) extends to the current setting and implies that society eventually samples from $\mathcal{S}^*$. This yields the proposition.\footnote{This proof also suggests that how small $\epsilon$ (and how large $B$) need to be depends on the distance between the optimal speed of learning and the ``second-best" speed of learning from any other minimal spanning set. Intuitively, in order to achieve long-run efficient learning, agents need to allocate $B$ observations in the best set to approximate the optimal frequencies. If another set of signals offers a speed of learning that is only slightly worse, we will need $B$ sufficiently large for the approximately optimal frequencies in the best set to beat this other set.}

\subsection{Proof of Proposition \ref{prop:freeSignal}}\label{appx:freeSignal}
Suppose without loss that the best set is $\{1, \dots, k\}$. By taking a linear transformation, we can further assume each signal $i$ with $1 \leq i \leq k$ only involves $\omega$ and the first $k-1$ confounding terms $b_1, \dots, b_{k-1}$. We claim that whenever $k-1$ sufficiently precise signals are provided about each of these confounding terms, long-run frequencies must converge to $\lambda^*$. 

Fix any positive real number $M$. Since the $k-1$ free signals are very precise, it is as if the prior precision matrix satisfies 
\[
(\Sigma^0)^{-1} \succeq M^2\sum_{i = 2}^{k}\Delta_{ii}
\]
where $\Delta_{ii}$ be the $K \times K$ matrix that has one at the $(i,i)$ entry and zero otherwise. Recall also that society eventually learns $\omega$. Thus at some late period $t_0$, society's acquisitions must satisfy 
\[
C'QC \succeq M^2 \Delta_{11}.
\]
Adding up the above two displays, we have 
\[
(\Sigma^0)^{-1} + C'QC \succeq M^2 \sum_{i = 1}^{k} \Delta_{ii} \succeq M c_j c_j', \forall 1 \leq j \leq k. 
\]
The last inequality uses the fact that each $c_j$ only involves the first $k$ coordinates. 

Now this is exactly the condition we need in order to apply Lemma \ref{lemmVarReduction}: Whether or not the condition is met for $j \notin \mathcal{S}^*$ does not affect the argument. Thus we can follow the proof of the restated Theorem \ref{thm:general} part (b) to deduce (\ref{eq:posteriorVarBound}). That is, for fixed $M$ and corresponding free information, society's long-run speed of learning cannot be slower than $(1+1/M) \cdot \phi(\mathcal{S}^*)^2$. This can be made arbitrarily close to the optimal speed, in which case we use Theorem \ref{thm:general} part (b) to conclude that society eventually samples according to $\lambda^*$.

\subsection{Proof of Proposition \ref{prop:highDelta}}\label{appx:highDelta}
Recall that $\lambda^*$ uniquely minimizes the function $V^*$. Thus the proposition is equivalent to the following: Fix $\epsilon > 0$, then for any $\delta$ close to $1$, $V(d_{\delta}(t)) \leq \frac{(1+\epsilon)\phi(\mathcal{S}^*)^2}{t}$ holds for sufficiently large $t$. That is, we only need to show that as $\delta \rightarrow 1$, the achieved speed of learning is close to the optimal speed.

Suppose for contradiction that $V(d_{\delta}(t)) > \frac{(1+\epsilon)\phi(\mathcal{S}^*)^2}{t}$ at some large $t$. Let $\tau < t$ be the last period with $V(d_{\delta}(\tau)) \leq \frac{(1+ \epsilon/2)\phi(\mathcal{S}^*)^2}{\tau}$. Below we first assume such a period $\tau$ exists; later we will show how to modify the proof when it does not. Consider the following deviation: 
\begin{enumerate}
\item Agents $i \leq \tau$ choose signals according to $d_{\delta}$ (i.e., they do not deviate); 
\item Starting in period $\tau+1$, the next $Mk$ agents sample each signal in the best set (of size $k$) exactly $M$ times, in an arbitrary order;
\item Starting in period $\tau+Mk+1$, each future agent chooses the signal that maximizes his own expected payoff, as in our main model. 
\end{enumerate}
In what follows we will show that for appropriately chosen $M$ as well as sufficiently large $\delta$ and $t$, this deviation yields a higher $\delta$-discounted payoff than the original strategy $d_{\delta}$. 

\bigskip

By construction, the deviation strategy achieves the same payoff as the original strategy in the first $\tau$ periods. Next we consider those periods $\tau + 1$ through $t$. For $1 \leq j \leq t - \tau$, let $\tilde{V}(\tau+j)$ denote the posterior variance at time $\tau+j$ under the deviation strategy. We can bound it from above as follows: Our previous analysis in Appendix \ref{appx:freeSignal} (specifically (\ref{eq:precisionIncrease})) gives that for $j > Mk$, 
\begin{equation}\label{eq:devPostVar1}
\frac{(M+1)\phi(\mathcal{S}^*)^2}{M \cdot \tilde{V}(\tau+j)} \geq \frac{(M+1)\phi(\mathcal{S}^*)^2}{M \cdot \tilde{V}(\tau+Mk)} + j-Mk,
\end{equation}
Using $\tilde{V}(\tau+Mk) \leq \tilde{V}(\tau) \leq \frac{(1+\epsilon/2)\phi(\mathcal{S}^*)^2}{\tau}$, the above inequality further yields
\begin{equation}\label{eq:devPostVar2}
\frac{(M+1)\phi(\mathcal{S}^*)^2}{M \cdot \tilde{V}(\tau+j)} \geq \frac{(M+1) \tau}{M (1+\epsilon/2)} + j - Mk.
\end{equation}
With slight algebra, we obtain from the above
\begin{equation}\label{eq:devPostVar3}
\frac{\tilde{V}(\tau+j)}{\phi(\mathcal{S}^*)^2} \leq \frac{1}{\frac{\tau}{1+\epsilon/2} + \frac{j-Mk}{1 + 1/M}}.
\end{equation}
Fixing $\epsilon$, we now choose $M$ so that $\frac{1}{M} < \frac{\epsilon}{4}$. Then there exists $\underline{j}$ (depending only on $\epsilon, M$ and $K$) such that for $j > \underline{j}$, it holds that 
\[
\frac{j - Mk}{1 + 1/M} \geq \frac{j+1}{1+\epsilon/2}.
\]
Thus, (\ref{eq:devPostVar3}) implies 
\begin{equation}\label{eq:devPostVar4}
\tilde{V}(\tau + j) \leq \frac{(1+\epsilon/2)\phi(\mathcal{S}^*)^2}{\tau + j +1}, ~~\forall \underline{j}+1 \leq j \leq t - \tau. 
\end{equation}
On the other hand, for small $j$ we have the following crude estimate: 
\begin{equation}\label{eq:devPostVar5}
\tilde{V}(\tau + j) \leq \tilde{V}(\tau) \leq \frac{(1+\epsilon/2)\phi(\mathcal{S}^*)^2}{\tau}, ~~\forall 1 \leq j \leq \underline{j}.
\end{equation}

Now we go back to the original strategy and make payoff comparisons. Our choice of $\tau$ ensures that posterior variance under the \emph{original} strategy is at least $\frac{(1+\epsilon/2)\phi(\mathcal{S}^*)^2}{\tau+j}$, for $1 \leq j \leq t - \tau$. Hence by deviating, the payoff gain in periods $\tau + 1 \sim t$ is at least\footnote{In this derivation we use (\ref{eq:devPostVar4}), (\ref{eq:devPostVar5}) as well as the identities $\frac{1}{\tau+j} - \frac{1}{\tau+j+1} = \frac{1}{(\tau+j)(\tau + j +1)}$ and $\frac{1}{\tau} - \frac{1}{\tau+j} = \frac{j}{\tau(\tau+j)}$.}
\[
\delta^{\tau} \cdot \underbrace{\left( \sum_{j = \underline{j}+1}^{t - \tau} \delta^{j-1} \left[\frac{(1+\epsilon/2)\phi(\mathcal{S}^*)^2}{(\tau+j)(\tau + j +1)} \right] - \sum_{ j = 1}^{\underline{j}} \delta^{j-1}\left[\frac{(1+\epsilon/2)\phi(\mathcal{S}^*)^2 j}{\tau(\tau+j)} \right]  \right)}_{(**)}.
\]
Note that $\underline{j}$ has already been fixed. So as $\delta \rightarrow 1$ and $t - \tau \rightarrow \infty$,\footnote{Since $\frac{(1+\epsilon)\phi(\mathcal{S}^*)^2}{t} \leq V(d_{\delta}(t)) \leq V(d_{\delta}(\tau)) \leq \frac{(1+\epsilon/2)\phi(\mathcal{S}^*)^2}{\tau}$, we have $\tau \leq \frac{1+\epsilon/2}{1+\epsilon} t$. So as $t$ becomes large, the difference $t - \tau$ also necessarily becomes large.} the term $(**)$ above converges to 
\[
(1+\epsilon/2)\phi(\mathcal{S}^*)^2 \cdot \left[ \frac{1}{\tau+\underline{j}+1} - \sum_{ j = 1}^{\underline{j}} \frac{j}{\tau(\tau+j)} \right].
\]
For large $\tau$, the above expression is larger than $\phi(\mathcal{S}^*)^2/\tau$.

Summarizing the above, we have shown that whenever $\tau > \underline{\tau}$, the deviation strategy achieves payoff gain in periods $\tau + 1$ through $t$ of at least $\delta^{\tau}\phi(\mathcal{S}^*)^2 / \tau$ (for $\delta$ close to 1). Although the deviation strategy might do worse in periods $t+1$ onwards, the potential payoff loss is at most $O(\frac{\delta^{t}}{1- \delta})$, which is smaller than the aforementioned payoff gain $\delta^{\tau}\phi(\mathcal{S}^*)^2 / \tau$ as $t - \tau \rightarrow \infty$  (since $\tau \leq \frac{1+\epsilon/2}{1+\epsilon} t$). Hence whenever $\tau > \underline{\tau}$, the deviation we constructed is a profitable deviation, and the proposition holds in these situations. 

\bigskip

Finally, we need to address the case where the previously-defined $\tau$ is weakly less than $\underline{\tau}$. This covers the case in which $\tau$ does not exist according to our earlier definition (simply let $\tau = 0$). Instead of (\ref{eq:devPostVar2}), we use the following weaker inequality
\[
\frac{(M+1)\phi(\mathcal{S}^*)^2}{M \cdot \tilde{V}(\tau+j)} \geq j - Mk.
\]
That is, $\tilde{V}(\tau+j) \leq \frac{(1+\frac{1}{M})\phi(\mathcal{S}^*)^2}{j-Mk}$. Recall that $\frac{1}{M} < \frac{\epsilon}{4}$ and $\tau \leq \underline{\tau}$ is now bounded. Thus for $j > \overline{j}$ (where $\overline{j}$ may need to be larger than $\underline{j}$), we would have 
\begin{equation}\label{eq:devPostVar6}
\tilde{V}(\tau+j) \leq \frac{(1+\epsilon/4)\phi(\mathcal{S}^*)^2}{\tau + j}, ~~ \forall \overline{j}+1 \leq j \leq t - \tau.
\end{equation}
And for small $j$ we can simply bound the posterior variance by the prior:
\begin{equation}\label{eq:devPostVar7}
\tilde{V}(\tau+j) \leq c, ~~ \forall 1 \leq j \leq \overline{j}.
\end{equation}
Using the estimates (\ref{eq:devPostVar6}) and (\ref{eq:devPostVar7}) in place of (\ref{eq:devPostVar4}) and (\ref{eq:devPostVar5}), we find that the deviation strategy achieves payoff gain in periods $\tau+1$ through $t$ of at least 
\[
\delta^{\tau} \cdot \left( \sum_{j = \overline{j}+1}^{t - \tau} \delta^{j-1} \left[\frac{\epsilon/4 \cdot \phi(\mathcal{S}^*)^2}{\tau+j} \right] - \sum_{ j = 1}^{\overline{j}} \delta^{j-1} c \right).
\]
Importantly, because (\ref{eq:devPostVar6}) improves upon (\ref{eq:devPostVar4}), we now have a harmonic sum in (the first part of) the parenthesis, which becomes arbitrarily large for $\delta$ close to $1$. Hence the above payoff gain is at least $\delta^{\tau}$ as $\delta \rightarrow 1$ and $t \rightarrow \infty$. Once again, this payoff gain dominates any potential loss after period $t$, showing that the deviation strategy is profitable. The proof of Proposition \ref{prop:highDelta} is complete.

\pagebreak
\begin{center}
\huge{For Online Publication}
\end{center}
\normalsize
\section{Online Appendix}

\subsection{Example Failing Unique Minimizer}\label{appx:examples}

There are $K=3$ states $\omega, b_1, b_2$ independently drawn with prior variances $\frac{1}{\alpha}, \frac{1}{\beta}, \frac{1}{\gamma}$. $N=4$ signals are available, and they are respectively 
\begin{align*}
X_1 &= \omega + b_1 + \epsilon_1 \\
X_2 &= b_1 + \epsilon_2 \\
X_3 &= \omega + b_2 + \epsilon_3 \\
X_4 &=  b_2 + \epsilon_4
\end{align*}
with standard normal errors. Note that the former two signals and the latter two signals are both spanning, and these two sets generate the same asymptotic variance. Thus Assumption \ref{assumptionUniqueMin} is not satisfied. 

The posterior variance about $\omega$ as a function of the number of observations $q_1,q_2,q_3,q_4$ of each signal type can be derived as follows. First, given $q_2$ observations of signal $X_2$ and $q_4$ observations of signal $X_4$, posterior variance about $\theta_2$ and $\theta_3$ are $1/(q_2+\beta)$ and $1/(q_4+\gamma)$ respectively. Consider now $q_1$ additional observations of $X_1$; this provides the same information about the payoff-relevant state $\omega$ as the signal $\omega + \epsilon'$, where $\epsilon'$ is an independent noise term with variance $ \frac{1}{q_1} + \frac{1}{q_2+\beta}$.  Similarly, $q_3$ additional observations of $X_3$ are equivalent to the signal $\omega + \epsilon''$, where $\epsilon''$ is an independent noise term with variance $ \frac{1}{q_3} + \frac{1}{q_4+\gamma}$. From this we deduce that posterior variance about $\omega$ is 
\[
V(q_1, q_2, q_3, q_4) = 1 \Bigg/ \left(\alpha+\frac{1}{\frac{1}{q_1} + \frac{1}{q_2+\beta}}+\frac{1}{\frac{1}{q_3} + \frac{1}{q_4+\gamma}} \right).
\]
The optimal division vector thus seeks to \emph{maximize} 
\begin{equation}\label{eq:q1q2q3q4}
\frac{1}{\frac{1}{q_1} + \frac{1}{q_2+\beta}}+\frac{1}{\frac{1}{q_3} + \frac{1}{q_4+\gamma}} 
\end{equation}
It is useful to rewrite (\ref{eq:q1q2q3q4}) in the following way:
\[\frac14 \left(q_1+q_2+\beta+q_3+q_4+\gamma - \frac{(q_1 - q_2 - \beta)^2}{q_1+q_2+\beta} - \frac{(q_3-q_4-\gamma)^2}{q_3+q_4+\gamma}\right).
\]
Then, since $q_1+q_2+\beta+q_3+q_4+\gamma=t+\beta+\gamma$ is fixed at any time $t$, it is equivalent to choose $q_1, q_2,q_3, q_4$ to minimize the sum of ratios
\[\frac{(q_1 - q_2 - \beta)^2}{q_1+q_2+\beta} + \frac{(q_3-q_4-\gamma)^2}{q_3+q_4+\gamma}.\]

\noindent Ideally,  if signals were perfectly divisible, the optimum would be to choose $q_1 = q_2 + \beta$ and $q_3 = q_4 + \gamma$. But as each $q_i$ is restricted to integer values, this continuous optimum is not feasible whenever $\beta$ and $\gamma$ are not integers.

The solution to this integer optimization problem is involved, and we need some additional notation. Let $r$ be the integer that minimizes $\lvert r - \beta \rvert$ (the distance to $\beta$) and let $s$ be the integer that minimizes $\lvert s - \gamma \rvert$. Further, let $\langle \beta \rangle$ and $\langle \gamma \rangle$ be the value of these absolute differences. 

\begin{claim} When the period $t$ has the same parity as $r+s$, the $t$-optimal $(q_1,q_2,q_3,q_4)$ satisfy
\[
q_1, q_2 \approx \frac{\langle\beta\rangle}{2\langle\beta\rangle + 2\langle\gamma\rangle} \cdot t; \quad q_3, q_4 \approx \frac{\langle\gamma\rangle}{2\langle\beta\rangle + 2\langle\gamma\rangle} \cdot t.
\] 
Otherwise the $t$-optimal $(q_1, q_2, q_3, q_4)$ satisfy
\[
q_1, q_2 \approx \frac{\langle\beta\rangle}{2\langle\beta\rangle + 2- 2\langle\gamma\rangle} \cdot t; \quad q_3, q_4 \approx \frac{1 - \langle\gamma\rangle}{2\langle\beta\rangle + 2 - 2 \langle\gamma\rangle} \cdot t.
\] 
\end{claim}
\bigskip
\noindent Thus, all four signals are observed with positive frequency in the long run according to the optimal criterion. 

\begin{proof} To solve the integer maximization problem (\ref{eq:q1q2q3q4}), let $r$ be the integer that minimizes $\lvert r - \beta \rvert$ (the distance to $\beta$) and let $s$ be the integer that minimizes $\lvert s - \gamma \rvert$. Further, let $\langle \beta \rangle$ and $\langle \gamma \rangle$ be the value of these absolute differences. We assume $2 \beta, 2 \gamma$ are not integers, so that $0 < \langle \beta \rangle, \langle \gamma \rangle < \frac{1}{2}$. We also assume $\langle \beta \rangle \neq \langle \gamma \rangle$, and by symmetry focus on the case of $\langle \beta \rangle < \langle \gamma \rangle$. 

With these assumptions, it is clear that when $q_1, q_2$ are integers, the minimum value of $\lvert q_1 - q_2 - \beta \rvert$ is $\langle \beta \rangle$, achieved if and only if $q_1 = q_2 + r$. Similarly the minimum value of $\lvert q_3 - q_4 - \gamma \rvert$ is $\langle \gamma \rangle$, achieved when $q_3 = q_4 + s$. Now if the total number of observations $t$ has the \emph{same parity} as $r + s$, it is possible to choose $q_1, q_2, q_3, q_4$ such that their sum is $t$ and $q_1 = q_2 + r$, $q_3 = q_4 + s$\textemdash any pair $q_2, q_4$ with sum $\frac{t - r - s}{2}$ leads to such a solution. Given these constraints, then, the optimum is to choose $q_2, q_4$ to minimize $\frac{\langle \beta \rangle^2}{2q_2+r+\beta} + \frac{\langle \gamma \rangle^2}{2q_4+s+\gamma}$. The optimal $q_2$ and $q_4$ satisfy $q_2/q_4 \approx \langle \beta \rangle / \langle \gamma \rangle$, which together with $q_2 + q_4 = \frac{t - r - s}{2}$ implies 
\[
q_1, q_2 \approx \frac{\langle\beta\rangle}{2\langle\beta\rangle + 2\langle\gamma\rangle} \cdot t; \quad q_3, q_4 \approx \frac{\langle\gamma\rangle}{2\langle\beta\rangle + 2\langle\gamma\rangle} \cdot t.
\] 

On the other hand, suppose $t$ has the \emph{opposite parity} to $r + s$. In this case $q_1 = q_2 + r$ and $q_3 = q_4 + s$ cannot both hold, thus $\lvert q_1 - q_2 - \beta \rvert$ and $\lvert q_3 - q_4 - \gamma \rvert$ cannot both take their minimum values $\langle \beta \rangle$ and $\langle \gamma \rangle$. It turns out that the best one can do is choose $q_1 = q_2 + r$ and $q_3 = q_4 + s \pm 1$ so that $\lvert q_1 - q_2 - \beta \rvert = \langle \beta \rangle$ and $\lvert q_3 - q_4 - \gamma \rvert = 1 - \langle \gamma \rangle$. Then, the optimal choice of $q_2, q_4$ with sum $\frac{t - r - s \mp 1}{2}$ to minimize $\frac{\langle \beta \rangle^2}{2q_2+r+\beta} + \frac{(1-\langle \gamma \rangle)^2}{2q_4+s+\gamma \pm 1}$. This yields
\[
q_1, q_2 \approx \frac{\langle\beta\rangle}{2\langle\beta\rangle + 2- 2\langle\gamma\rangle} \cdot t; \quad q_3, q_4 \approx \frac{1 - \langle\gamma\rangle}{2\langle\beta\rangle + 2 - 2 \langle\gamma\rangle} \cdot t\]
as desired.
\end{proof}

\bigskip

The intuition for the conclusion above is simple: We would most prefer to set $q_1 = q_2 + \beta$ and $q_3 = q_4 + \gamma$, but this is not feasible when $\beta$ and $\gamma$ are not integers. Thus, there is inevitably some loss from the ideal case where signals are perfectly divisible. This loss turns out to be convex in signal counts, so both groups of signals are observed infinitely often to minimize total loss.

\subsection{Strengthening of Theorem \ref{thm:general} part (b)} \label{appx:m(t)finite}
Here we show the following result, which strengthens the restated Theorem \ref{thm:general} part (b) (see Appendix \ref{appx:myopicConvergence}). It implies that under Assumption \ref{assumptionUniqueMinStrong}, any signal that is observed with zero long-run frequency must in fact be observed only finitely often. 

\bigskip

\noindent \textbf{Stronger Version of Theorem \ref{thm:general} part (b).} \emph{Suppose that the signals observed infinitely often span $\mathbb{R}^{K}$. Then $m_i(t) = \lambda^*_i \cdot t + O(1), \forall i$. }

\bigskip

The proof is divided into two subsections below.

\subsubsection{A Weaker Result}
Recall that we have previously shown $m_i(t) \sim \lambda^*_i \cdot t$. We can first improve the estimate of the residual term to $m_i(t) = \lambda^*_i \cdot t + O(\ln t)$. Indeed, Lemma \ref{lemmVarReduction} yields that for some constant $L$ and every $t \geq L$, 
\begin{equation}\label{eq4}
V(m(t+1)) \leq V(m(t)) - \left(1- \frac{L}{t}\right) \cdot \frac{V(m(t))^2}{\phi(\mathcal{S^*})^2}. 
\end{equation}
This is because we may apply Lemma \ref{lemmVarReduction} with $M = \min_{j \in \mathcal{S}^*} m_j(t)$, which is at least $\frac{t}{L}$. 

Let $g(t) = \frac{V(m(t))}{\phi(\mathcal{S}^*)^2}$. Then the above simplifies to
\[
g(t+1) \leq g(t) - \left(1 - \frac{L}{t}\right) g(t)^2. 
\]
Inverting both sides, we have
\begin{equation}\label{eq5}
\frac{1}{g(t+1)} \geq \frac{1}{g(t)} + \frac{1 - L/t}{1 - (1-L/t)g(t)} \geq \frac{1}{g(t)} + 1 - \frac{L}{t}. 
\end{equation}
This enables us to deduce 
\[
\frac{1}{g(t)} \geq \frac{1}{g(L)} + \sum_{x = L}^{t - 1}\left(1 - \frac{L}{x}\right) \geq t - O(\ln t).
\]
Thus $g(t) \leq \frac{1}{t - O(\ln t)} \leq \frac{1}{t} + O(\frac{\ln t}{t^2})$. That is,
\[
V(m(t)) \leq \frac{\phi(\mathcal{S}^*)^2}{t} + O\left(\frac{\ln t}{t^2}\right). 
\]
Since $t \cdot V(\lambda t)$ approaches $V^*(\lambda)$ at the rate of $\frac{1}{t}$, we have 
\begin{equation}\label{eq6}
V^*\left(\frac{m(t)}{t}\right) \leq t \cdot V(m(t)) + O\left(\frac{1}{t}\right) \leq \phi(\mathcal{S}^*)^2 + O\left(\frac{\ln t}{t}\right). 
\end{equation}
Suppose $\mathcal{S}^* = \{1, \dots, k\}$. Then the above estimate together with (\ref{eq:V*lowerBound}) implies $\sum_{j > k} \frac{m_j(t)}{t} = O(\frac{\ln t}{t})$. Hence $m_j(t) = O(\ln t)$ for each signal $j$ outside of the best set. 

Now we turn attention to those signals in the best set. If these were the only available signals, then the analysis in \citet{LiangMuSyrgkanis} gives $\partial_i V(m(t)) = -\left(\frac{\beta_i^{\mathcal{S}^*}}{m_i(t)}\right)^2$. In our current setting, signals $j > k$ affect this marginal value of signal $i$, but the influence is limited because $m_j(t) = O(\ln t)$. Specifically, we can show that
\[
\partial_i V(m(t)) = - \left(\frac{\beta_i^{\mathcal{S}^*}}{m_i(t)}\right)^2 \cdot \left(1 + O\left(\frac{\ln t}{t}\right)\right).
\]
This then implies $m_i(t) \leq \lambda^*_i \cdot t + O(\ln t)$.\footnote{Otherwise, consider $\tau+1 \leq t$ to be the last period in which signal $i$ was observed. Then $m_i(\tau)$ is larger than $\lambda^*_i \cdot \tau$ by several $\ln (\tau)$, while there exists some other signal $\hat{i}$ in the best set with $m_{\hat{i}}(\tau) < \lambda^*_i \cdot \tau$. But then $\lvert \partial_i V(m(\tau)) \rvert < \lvert \partial_{\hat{i}} V(m(\tau)) \rvert$, meaning that the agent in period $\tau+1$ should not have chosen signal $i$.} Using $\sum_{i \leq k} m_i(t) = t - O(\ln t)$, we deduce that $m_i(t) \geq \lambda^*_i \cdot t - O(\ln t)$ must also hold. Hence $m_i(t) = \lambda^*_i \cdot t + O(\ln t)$ for each signal $i$. 

\subsubsection{Getting Rid of the Log}
In order to remove the $\ln t$ residual term, we need a refined analysis. The reason we ended up with $\ln t$ is because we used (\ref{eq4}) and (\ref{eq5}) at \emph{each} period $t$; the ``$\frac{L}{t}$" term in those equations adds up to $\ln t$. In what follows, instead of quantifying the variance reduction in each period (as we did), we will lower-bound the variance reduction over multiple periods. This will lead to better estimates and enable us to prove $m_i(t) = \lambda^*_i \cdot t + O(1)$. 

To give more detail, let $t_1 < t_2 < \dots$ denote the periods in which some signal $j > k$ is chosen. Since $m_j(t) = O(\ln t)$ for each such signal $j$, $t_l \geq 2^{\epsilon \cdot l}$ holds for some positive constant $\epsilon$ and each positive integer $l$. Continuing to let $g(t) = \frac{V(m(t))}{\phi(\mathcal{S}^*)^2}$, our goal is to estimate the difference between $\frac{1}{g(t_{l+1})}$ and $\frac{1}{g(t_{l})}$. 

Ignoring period $t_{l+1}$ for the moment, we are interested in $\frac{\phi(\mathcal{S}^*)^2}{V(m(t_{l+1}-1))} - \frac{\phi(\mathcal{S}^*)^2}{V(m(t_{l}))}$, which is just the difference in the \emph{precision} about $\omega$ when the division vector changes from $m(t_l)$ to $m(t_{l+1}-1)$. From the proof of Lemma \ref{lemmDirectionalDer}, we can estimate this difference if the change were along the direction $\lambda^*$:
\begin{equation}\label{eq7}
\frac{\phi(\mathcal{S}^*)^2}{V(m(t_l) + \lambda^*(t_{l+1}-1-t_l))} - \frac{\phi(\mathcal{S}^*)^2}{V(m(t_{l}))} \geq t_{l+1} - 1 - t_l.
\end{equation}
Now, the vector $m(t_{l+1}-1)$ is not exactly equal to $m(t_l) + \lambda^*(t_{l+1}-1-t_l)$, so the above estimate is not directly applicable. However, by our definition of $t_l$ and $t_{l+1}$, any difference between these vectors must be in the first $k$ signals. In addition, the difference is bounded by $O(\ln  t_{l+1})$ by what we have shown. This implies\footnote{By the mean-value theorem, the difference can be written as $O(\ln  t_{l+1})$ multiplied by a certain directional derivative. Since the coordinates of $m(t_{l+1}-1)$ and of $m(t_l) + \lambda^*(t_{l+1}-1-t_l)$ both sum to $t_{l+1}-1$, this directional derivative has a direction vector whose coordinates sum to zero. Combined with $\partial_i V(m(t)) = - (\frac{\phi(\mathcal{S}^*)^2}{t}) \cdot (1 + O(\frac{\ln t}{t}))$ (which we showed before), this directional derivative has size $O(\frac{\ln t}{t^3})$. 
}
\[
V(m(t_{l+1}-1)) - V(m(t_l) + \lambda^*(t_{l+1}-1-t_l))= O\left (\frac{\ln^2 t_{l+1}}{t_{l+1}^3} \right).
\]
Since $V(m(t_{l+1}-1))$ is on the oder of $\frac{1}{t_{l+1}}$, we thus have (if the constant $L$ is large)
\begin{equation}\label{eq8}
\frac{\phi(\mathcal{S}^*)^2}{V(m(t_{l+1}-1))}-\frac{\phi(\mathcal{S}^*)^2}{V(m(t_l) + \lambda^*(t_{l+1}-1-t_l))} \geq -\frac{L \ln^2t_{l+1}}{t_{l+1}}. 
\end{equation}
(\ref{eq7}) and (\ref{eq8}) together imply 
\[
\frac{1}{g(t_{l+1}-1)} \geq \frac{1}{g(t_l)} + (t_{l+1} - 1 - t_l) - \frac{L \ln^2t_{l+1}}{t_{l+1}}.
\]
Finally, we can apply (\ref{eq5}) to $t = t_{l+1} - 1$. Altogether we deduce
\[
\frac{1}{g(t_{l+1})} \geq \frac{1}{g(t_l)} + (t_{l+1} - t_l) - \frac{2L \ln^2t_{l+1}}{t_{l+1}}. 
\]

Now observe that $\sum_{l} \frac{2L \ln^2t_{l+1}}{t_{l+1}}$ converges (this is the sense in which our estimates here improve upon (\ref{eq5}), where $\frac{L}{t}$ leads to a divergent sum). Thus we are able to conclude
\[
\frac{1}{g({t_l})} \geq t_l - O(1), \quad \forall l. 
\]
In fact, this holds also at periods $t \neq t_l$. Therefore $V(m(t)) \leq \frac{\phi(\mathcal{S}^*)^2}{t} + O(\frac{1}{t^2})$, and 
\begin{equation}\label{eq9}
V^*\left(\frac{m(t)}{t}\right) \leq t \cdot V(m(t)) + O\left(\frac{1}{t}\right) \leq \phi\left(\mathcal{S}^*\right)^2 + O\left(\frac{1}{t}\right).
\end{equation}
This equation (\ref{eq9}) improves upon the previously-derived (\ref{eq6}). Hence by (\ref{eq:V*lowerBound}) again, $m_j(t) = O(1)$ for each signal $j > k$. And once these signal counts are fixed, Proposition \ref{prop:N=K} implies $m_i(t) = \lambda^*_i \cdot t + O(1)$ also holds for signal $i \leq k$. This completes the proof. 

\subsection{Example of a Learning Trap with Non-Normal Signals} \label{appx:non-normalLT}
The payoff-relevant state $\theta \in \{\theta_1, \theta_2\}$ is binary and agents have a uniform prior. There are three available information sources. The first, $X_1$, is described by the information structure
\[
\begin{array}{ccc}
& \theta_1 & \theta_2 \\
s_1 & p & 1-p \\
s_2 & 1-p & p
\end{array}
\]
with $p>1/2$. Information sources 2 and 3 provide perfectly correlated signals (conditional on $\theta$) taking values in $\{a,b\}$: In state $\theta_1$, there is an equal probability that $X_2=a$ and $X_3=b$ or $X_2=b$ and $X_3=a$. In state $\theta_2$, there is an equal probability that $X_2=X_3=a$ and $X_2=X_3=b$. 

In this environment, every agent chooses to acquire the noisy signal $X_1$, even though one observation of each of $X_2$ and $X_3$ would perfectly reveal the state.\footnote{We thank Andrew Postlewaite for this example.}

\subsection{Example Mentioned in Section \ref{sec:preciseInfo}} \label{appx:preciseInfo}
Suppose the available signals are 
\begin{align*}
X_1 &= 10 x + \epsilon_1 \\
X_2 &= 10 y + \epsilon_2 \\
X_3 &= 4 x + 5 y + 10 b \\
X_4 &= 8 x + 6 y -20 b
\end{align*}
where $\omega = x+y$ and $b$ is a payoff-irrelevant unknown. Set the prior to be
\[(x,y,b)'\sim \mathcal{N}\left(\left(\begin{array}{c} 0 \\ 0 \\ 0 \end{array}\right), \left(\begin{array}{ccc}
0.1 & 0 & 0 \\
0 & 0.1 & 0 \\
0 & 0 & 0.039\end{array}\right)\right).
\]

It can be computed that agents observe only the signals $X_1$ and $X_2$, although the set $\{X_3,X_4\}$ is optimal with $\phi(\{X_1,X_2\})=1/5>3/16=\phi(\{X_3,X_4\})$. Thus, the set $\{X_1, X_2\}$ constitutes a learning trap for this problem. But if each signal choice were to produce ten independent realizations, agents starting from the above prior would observe only the signals $X_3$ and $X_4$. This breaks the learning trap.

\subsection{Supplementary Material to Section \ref{sec:extensions}} 

\subsubsection{General Payoff Functions} \label{appx:genpayoff}
We comment here on the possibilities for (and limitations to) generalizing Proposition \ref{prop:highDelta} beyond the quadratic loss payoff function. Proposition \ref{prop:highDelta} does extend to some other ``prediction problems,'' in which every agent's payoff function $u(a, \omega)$ is the same and depends only on $\lvert a - \omega \rvert$. For example, the proposition holds for $u(a, \omega) = \lvert a - \omega \rvert^{\gamma}$ with any exponent $\gamma \in (0, 2]$; such an extension only requires minor changes to our proof in Appendix \ref{appx:highDelta}. 

Nonetheless, even restricting to prediction problems, Proposition \ref{prop:highDelta} does \emph{not} hold in general. For a counterexample, consider $u(a, \omega) = -\mathbf{1}_{\{\lvert a - \omega \rvert > 1\}}$, which punishes the agent for any prediction that differs from the true state by more than $1$.\footnote{We thank Alex Wolitzky for this example.} Intuitively,  the payoff gain from further information decreases sharply (indeed, exponentially) with the amount of information that has already been acquired. Thus, even with a forward-looking objective function, the range of future payoffs is limited and each agent cares mostly to maximize his own payoff. This results in an optimal procedure that resembles myopic behavior (and differs from the procedure that would maximize speed of learning).

The above counterexample illustrates the difficulty in estimating the value of information when working with an arbitrary payoff function. In order to make intertemporal payoff comparisons, we need to know how much payoff is gained/lost when the posterior variance is decreased/increased by a certain amount.\footnote{This difficulty becomes more salient if we try to go beyond prediction problems: The value of information in that case will depend on signal realizations, which brings another challenge of stochasticity.} This can be challenging in general, see \citet{ChadeSchlee} for a related discussion.

Finally, while it is more than necessary to assume that agents have the same payoff function, the truth of Proposition \ref{prop:highDelta} does require some restrictions on how the payoff functions differ. Otherwise, suppose for example that payoffs take the form $-\alpha_t (a_t-\omega)^2$, where $\alpha_t$ decreases exponentially fast. Then even with the $\delta$-discounted objective, the social planner puts most of the weight on earlier agents, resembling the myopic behavior of individual agents. 

\subsubsection{Low Altruism} \label{appx:lowDelta}
Here we argue that part (a) of Theorem \ref{thm:general} generalizes to agents who are not completely myopic, but are sufficiently impatient. That is, we will show that if signals $1, \dots, k$ are subspace-optimal, then there exist priors given which agents with low $\delta$ always observe these signals in equilibrium. 

We follow the construction in Appendix \ref{appx:myopicExistence}. The added difficulty here is to show that if any agent ever chooses a signal $j > k$, the payoff loss in that period (relative to myopically choosing among the first $k$ signals) is of the same magnitude as possible payoff gains in future periods. Then, for sufficiently small $\delta$, such a deviation is not profitable. 

Suppose that agents sample only from the first $k$ signals in the first $t-1$ periods, with frequencies close to $\lambda^*$. Then, the posterior variances $v_1, \dots, v_k$ (which are also the prior for period $t$) are on the order of $\frac{1}{t}$. Thus any signal acquisition in period $t$ leads to a variance reduction on the order of $\frac{1}{t^2}$. However, using the computation in Appendix \ref{appx:myopicExistence}, we can show that for some positive constant $\xi$ (independent of $t$), the variance reduction of any signal $j > k$ is smaller than the variance reduction of any of the first $k$ signals by $\frac{\xi}{t^2}$. This is the amount of payoff loss in period $t$ under a deviation to signal $j$. 

Of course, this deviation could improve the posterior variance in future periods. But even for the best continuation strategy, the posterior variance in period $t + m$ can at most be reduced by $O(\frac{m}{t^2})$.\footnote{This is because over $m$ periods, the increase in the \emph{precision} matrix is $O(m)$.} Thus if we choose $\xi$ to be small enough, the payoff gain in period $t+m$ is bounded above by $\frac{m}{\xi t^2}$. Note that for $\delta$ sufficiently small, 
\[
- \frac{\xi}{t^2} + \sum_{m \geq 1} \delta^m \cdot \frac{m}{\xi t^2} < 0.
\]
Hence the deviation is not profitable and the proof is complete. 

\subsubsection{Multiple Payoff-Relevant States} \label{appx:multiStates}
Let $V(q_1,\dots,q_N)$ be a weighted sum of  posterior variances about the $r$ payoff-relevant states. As before, define $V^*$ to be a normalized, asymptotic version of $V$. Let $n(t)$ continue to represent any allocation of $t$ observations that minimizes $V$. Then, under a modification of the Unique Minimizer assumption\textemdash we require $V^*$ to be \emph{uniquely} minimized\textemdash the optimal frequency vector $\lambda^{OPT}:= \lim_{t \rightarrow \infty} n(t)/t$ is well-defined. Nonetheless, we emphasize that with $r > 1$, these optimal allocations generally involve more than $K$ signals. A theorem of \citet{Chaloner} shows that $\lambda^{OPT}$ is supported on at most $\frac{r(2K+1-r)}{2}$ signals.

We can generalize the notion of ``minimal spanning sets" as follows: A set of signals $\mathcal{S}$ is minimally-spanning if optimal sampling from $\mathcal{S}$ puts positive frequency on \emph{every} signal in $\mathcal{S}$. When $r = 1$, this definition agrees with the definition in our main model. But for $r > 1$, we no longer know of a simple method for checking whether a set is minimally-spanning. 

Similarly, we say that a minimal spanning set $\mathcal{S}$ is ``subspace-optimal" if, when agents are constrained to choose from $\overline{\mathcal{S}}$, the optimal frequency vector is supported on $\mathcal{S}$. With these definitions, Theorem \ref{thm:general} and its proof generalize without change.

\end{document}